%% file: arxiv.tex
\documentclass[sigconf, noacm]{acmart}
\AtBeginDocument{%
  }

\usepackage{algorithm}
\usepackage[noend]{algpseudocode}
\usepackage{subcaption}
\usepackage{balance} 
\usepackage{thmtools}
\usepackage{thm-restate}

\setcopyright{none}             
\renewcommand\footnotetextcopyrightpermission[1]{} 
\usepackage{fancyhdr}

\begin{document}

\title{Submodular Maximization under Supermodular Constraint: Greedy Guarantees}
\author{Ajitesh Srivastava}
\authornote{Part of the work was done while affiliated with the University of Southern California.}
\email{aj.srivastava@northeastern.edu}
\affiliation{%
  \institution{Northeastern University}
  \city{Charlotte}
  \state{North Carolina}
  \country{USA}
}

\author{Shanghua Teng}
\email{shanghua@usc.edu}
\affiliation{%
  \institution{University of Southern California}
  \city{Los Angeles}
  \state{California}
  \country{USA}}



\begin{abstract}

Motivated by a wide range of applications in data mining and machine learning, we consider the problem of maximizing a submodular function subject to supermodular cost constraints. In contrast to the well-understood setting of cardinality and matroid constraints, where greedy algorithms admit strong guarantees, the supermodular constraint regime remains poorly understood -- guarantees for greedy methods and other efficient algorithmic paradigms are largely open. We study this family of fundamental optimization problems under an upper-bound constraint on a supermodular cost function with curvature parameter $\gamma$.  Our notion of supermodular curvature is less restrictive than prior definitions, substantially expanding the class of admissible cost functions.
We show that our greedy algorithm, which iteratively includes elements maximizing the ratio of the objective and constraint functions, achieves a $\left(1 - e^{-(1-\gamma)}\right)$-approximation before stopping. We prove that this approximation is indeed tight for this algorithm. Further, if the objective function has a submodular curvature $c$, then we show that the bound further improves to $\left(1 - (1- (1-c)(1-\gamma))^{1/(1-c)}\right)$, which can be further improved by continuing to violate the constraint. Finally, we show that the Greedy-Ratio-Marginal in conjunction with binary search leads to a bicriteria approximation for the dual problem -- minimizing a supermodular function under a lower bound constraint on a submodular function. We conduct a number of experiments on a simulation of LLM agents debating over multiple rounds -- the task is to select a subset of agents to maximize correctly answered questions. Our algorithm outperforms all other greedy heuristics, and on smaller problems, it achieves the same performance as the optimal set found by exhaustive search.

\end{abstract}


\begin{teaserfigure}
  \includegraphics[width=\textwidth]{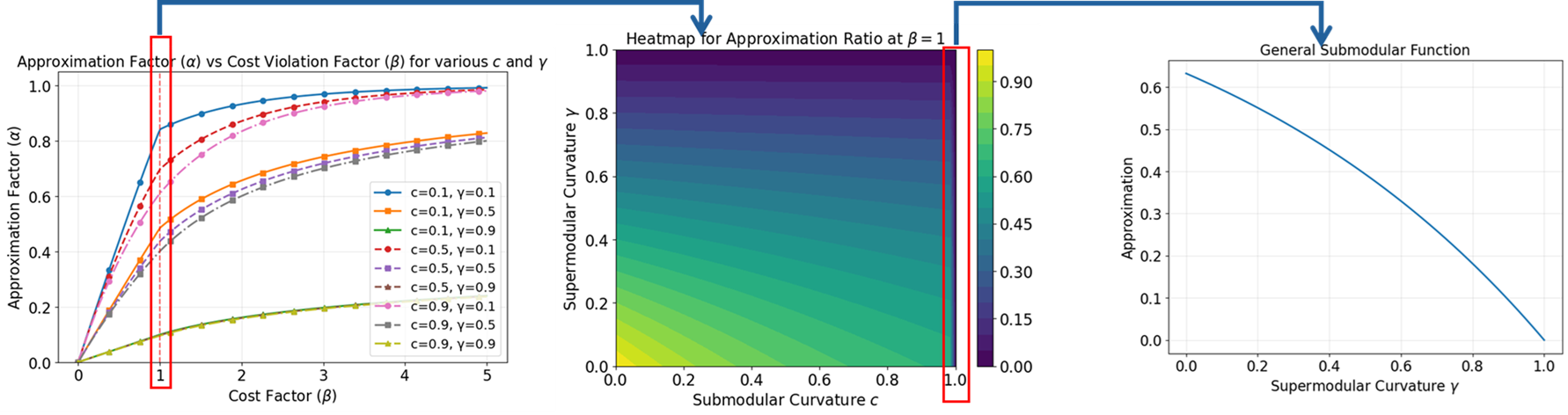}
  \caption{Summary of our results for submodular maximization with supermodular constraint: We find the approximation factor as a function of the supermodular curvature (right) when Greedy algorithm exceeds the constraint for the first time. We obtain better approximation when submodular objective has a known curvature (middle). We further find how approximation factor improves if we continue to overflow the constraint by a factor $\beta$ (left).}
  \label{fig:teaser}
\end{teaserfigure}


\pagestyle{plain}
\fancyhf{} 
\cfoot{\thepage} 

\maketitle

\section{Introduction}

The greedy algorithm is one of the most intuitive and widely used algorithmic paradigms. Owing to its simplicity and scalability, it is ubiquitous in practice, ranging from machine learning applications such as feature selection~\cite{chandrashekar2014survey} to real-world resource management problems such as vaccine allocation~\cite{islam2021evaluation}.
Despite its simplicity, greedy algorithms achieve provably strong approximation guarantees for a broad class of combinatorial optimization problems.

A particularly notable application is submodular maximization under cardinality, linear, or matroid constraints, with prominent examples including influence maximization~\cite{kempe2003maximizing} and recommender systems. In these classical settings, greedy selection yields constant-factor approximations for otherwise intractable optimization problems and is known to be optimal in the value-oracle model.

However, a growing class of submodular maximization problems arising in data mining and machine learning—including feature or rule selection for explanatory coverage, as well as agent selection in LLM-based debating frameworks with interdependent costs—involves supermodular cost constraints (more discussion see Section \ref{subsection:SOSCinML}), which fundamentally challenge existing approximation theory. Such constraints naturally arise from higher-order interaction, interference, or coordination effects among selected elements, leading to superlinear cost growth.


\subsection{Algorithmic Characterizations}

Submodular maximization under a supermodular constraint (SMSC) concerns the following optimization problem.

\begin{eqnarray}
    \max_{S\subseteq V} \; f(S) \quad\text{subject to}\quad g(S) \le \theta
\end{eqnarray}
where 
$f:2^V\to\mathbb{R}_{\ge0}$ and $g:2^V\to\mathbb{R}_{\ge0}$ are, respectively, 
\textit{normalized, monotone, submodular} and \textit{normalized, monotone, supermodular} functions over a finite ground set $V$, and $\theta>0$ is a budget.


For monotone submodular maximization subject to a non-linear (supermodular) knapsack constraint, no polynomial-time constant-factor approximation is possible in the black-box value-oracle model without additional structural assumptions on the cost function. To further highlight the inherent difficulty of this optimization problem, we observe the following hardness result via a simple reduction from the {\sc Maximum Independent Set} problem~\cite{haastad1999clique}.

\begin{proposition}[Approximation Intractability]
    Even maximizing a modular function subject to a non-linear supermodular knapsack constraint cannot be approximated within any constant factor in polynomial time, unless NP = P.
    This result also implies a lower bound of superpolynomially many queries in the black-box value-oracle model.
\end{proposition}

The reduction can be performed by considering the size of the desired independent set $S$ as the (sub)modular objective $f(S)$ and setting the constraint as some constant ($\theta$) plus the number of edges induced by the set $|E_S|$. Note that the constraint is supermodular as $|E_S|$ is supermodular~\cite{chekuri2022densest}.
Then the solution that maximizes $f(S)$ such that $g(S) \geq \theta$ must pick a set with $|E_S| = 0$, and thus $S$ would be the largest independent set. Therefore, any hope for a polynomial-time approximation of this fundamental optimization problem necessitates additional structural assumptions on the supermodular constraint function.

Motivated by several emerging applications in machine learning and data mining, we consider the problem of maximizing a submodular function $f$subject to an upper-bound constraint on a supermodular cost function $g$ with curvature 
$\gamma \in (0,1)$.

We define the curvature $\gamma$ of a supermodular function $g$ as:

\begin{eqnarray}\label{Eqn:SupermodularCurvature}
    \gamma = 1 - \min_{g(T)\geq g(S)} \frac{g(T)}{g(T \mid S)}.
\end{eqnarray}
Our notion of supermodular curvature is strictly less restrictive than previously studied notions of supermodular curvature~\cite{bai2018greed},
where the minimum is taken over $g(v)/g(v \mid S)$ over all elements $v$ and sets $S$ (see Section~\ref{sec:sub-sup} for details).
This substantially expands the class of admissible cost functions to include, for example, quadratic costs arising from interference or coordination among selected elements.
Intuitively, smaller curvature corresponds to behavior closer to linear: when $\gamma = 0$, the function $g$ is modular (linear), whereas as $\gamma$ approaches $1$, the ``synergy" or ``interference" effects in the cost function become increasingly pronounced. See more discussions of these potential applications in Section \ref{subsection:SOSCinML}.

We analyze a greedy algorithm that iteratively selects the element maximizing the ratio between marginal objective gain and marginal constraint cost, terminating at the first violation of the upper-bound constraint.
We prove:

\begin{restatable}[Greedy Overflow Characterization]{theorem}{thmBoundary}\label{thmBoundary}
    If our greedy algorithm at step $k-1$ is within the constraint and at $k$ overflows the constraint by a factor of $\beta > 1$, i.e., $\beta = g(S_k)/\theta$, then at least one of the following is true.
   
    \noindent (i) $f(S_{k-1}) \geq (1 - e^{-(1-\gamma)}) f(S^*)$, or
    (ii) $f(S_k) \geq \left(1-e^{-\beta(1-\gamma)}\right) f(S^*)$,

\noindent    where $S^*$ is the optimal set.
\end{restatable}
Since the overflow is bounded by 
$\beta \leq (2-\gamma)/(1-\gamma)$, Theorem~\ref{thmBoundary} yields a bicriteria guarantee. We further show that this characterization is tight for the greedy algorithm (Theorem \ref{thm:tightness}) and establish bicriteria approximation guarantees for any truncation of the greedy process. Moreover, when the objective 
$f$ has bounded submodular curvature, our guarantees can be further strengthened. 
Finally, we consider, for a target value $\tau>0$, the dual supermodular minimization subject to a submodular coverage constraint:
\begin{eqnarray}
    \min_{S\subseteq V} \; g(S) \quad\text{subject to}\quad f(S) \ge \tau.
\end{eqnarray}
 
In Theorem~\ref{thmDual}, we present an efficient binary search reduction that converts any algorithm with bicriteria guarantees for primal submodular maximization 
into an algorithm with bicriteria guarantees for the dual supermodular minimization problem.
  



\subsection{Submodular Objectives and Supermodular Constraints in Machine Learning}
\label{subsection:SOSCinML}

By providing a tight characterization of the performance of greedy methods for this fundamental family of submodular/supermodular optimization problems—particularly in terms of a less restrictive notion of supermodular curvature—we aim to capture the algorithmic potential of these methods for data-driven applications.

Submodularity is a foundational concept in machine learning, underpinning a wide range of tasks in selection, summarization, and interpretability. For example, in interpretable decision rule learning~\cite{yang2021learning}, the objective is formulated as a regularized submodular maximization problem that balances predictive coverage (a monotone submodular function) with modular penalties for complexity and overlap. The resulting distorted greedy algorithms achieve strong $(1-1/e)$
 approximation guarantees while scaling efficiently to high-dimensional data.
In image attribution, submodular subset selection~\cite{chen2024less} enables the identification of compact, diverse, and semantically meaningful regions that drive model predictions, with objectives explicitly combining confidence, diversity, consistency, and collaboration scores—all shown to be submodular. In this setting, greedy optimization not only improves interpretability but also outperforms both black-box and white-box baselines on standard evaluation metrics.
More broadly, frameworks such as PRISM~\cite{kothawade2022prism} introduce parameterized submodular information measures, including Submodular Mutual Information (SMI) and Submodular Conditional Gain (SCG), to guide data selection for targeted learning, privacy-preserving summarization, and rare class discovery. These approaches leverage the theoretical guarantees and scalability of submodular maximization, enabling applications ranging from active learning to large-scale data summarization. Recent advances further extend submodular optimization to dynamic action space selection in large language models~\cite{zhao2025dynaact}, where submodular objectives promote both utility and diversity in candidate actions.

However, much of the existing literature relies on simplifying constraints to cardinality bounds or linear (modular) cost functions in order to exploit known approximation results. In contrast, many practical constraints are inherently non-linear and exhibit supermodular behavior, in which the marginal cost of adding an element increases with the size of the selected set. The following examples from machine learning and other realistic applications illustrate that supermodular cost structures arise naturally across diverse domains, motivating the need for principled approximation guarantees beyond linear or cardinality constraints.


\textbf{Active Learning with Batch Costs:}
Selecting a batch of samples for human labeling naturally induces a submodular objective, as adding new examples to an already large set yields diminishing returns. Meanwhile, the time cost incurred by an expert to label additional samples may grow superlinearly when cross-referencing is required to ensure label consistency.

 \textbf{Feature Selection with Model Costs:}
Selecting features to maximize a submodular information-gain objective subject to a model cost budget—measured, for example, by computational cost or representational capacity—may incur superlinear costs as the number of selected features increases.

\textbf{Sensor Network Design with Interference:}
When deploying a sensor network to maximize coverage, adding a sensor can incur supermodular costs due to increased communication interference or data synchronization overhead in dense deployments.

\textbf{Debating Agents for Factuality and Reasoning:}
Multi-agent debate frameworks improve the reliability and factuality of large language models by coordinating multiple agents to iteratively propose, critique, and refine candidate answers. This process enforces cross-verification among agents, reducing hallucinations and improving reasoning performance. 
Empirical studies show that increasing the number of agents or debate rounds leads to better performance on tasks such as mathematical problem solving and visual question answering. The approach is plug-and-play, compatible with diverse LLMs, and demonstrates the power of collective intelligence in AI systems.
Increasing the number of agents or debate rounds typically yields diminishing gains in accuracy while incurring superlinear coordination and communication costs, naturally giving rise to supermodular cost structures.

The dual variant—minimizing a supermodular cost subject to a submodular coverage requirement—is equally relevant. For example, in privacy-preserving data summarization, one may seek to minimize cumulative privacy risk, which can be supermodular when combinations of data points reveal more information than the sum of their individual contributions, while ensuring that the selected summary attains a target level of utility or representativeness. More generally, any of the above problems can be reformulated to emphasize cost minimization under a 
coverage constraint.

\subsection{Approximation Results in Similar Problems}
Submodular function optimization has been extensively studied. For maximization
under a cardinality constraint, the simple greedy algorithm achieves a ratio of $(1-1/e) \approx 0.632$, which is optimal under $P \ne NP$~\cite{nemhauser1978analysis,kempe2003maximizing}; for a modular knapsack constraint, an augmented greedy algorithm achieves the same $(1-1/e)$ guarantee.
Tigher bound can be achieved when considering 
curvature
\cite{ModularityCurvatureSviridenko}.

When constraints become non-modular, the problem becomes significantly more difficult. Maximizing submodular functions under a submodular constraint is not known to admit constant-factor approximations. This problem can be seen as the complement of ours. It considers an upper bound on a submodular function as the constraint, which can be transformed into a lower bound on a supermodular function. 
On this problem, existing work has achieved approximation with data-dependent additive error~\cite{padmanabhan2023maximizing}.  
Alternatively, constant-factor bicriteria approximations have been obtained for formulations of the problem with surrogate functions that approximate the original submodular constraint~\cite{chen2023bicriteria}.
The sum of submodular and supermodular functions has also been studied under cardinality constraints when both functions have bounded curvature~\cite{bai2018greed,SHI2025113}.
To the best of our knowledge, there are no prior constant approximation guarantees (bicriteria or otherwise) for maximizing a monotone submodular function under a supermodular cost constraint.
Bicriteria approximations, which satisfy the objective target while violating the budget, provide a powerful framework for obtaining constant-factor guarantees in these challenging settings.


\section{Background}

\label{sec:sub-sup}
A set function $f:2^V \to \mathbb{R}$ is \textit{submodular} if for every $A \subseteq B \subseteq V$ and $e \in V \setminus B$, it satisfies the property of diminishing marginal returns: $f(A \cup \{e\}) - f(A) \ge f(B \cup \{e\}) - f(B).$
Equivalently, $f$ is submodular if for all $S, T \subseteq V$, $f(S \cup T) + f(S \cap T) \le f(S) + f(T)$. 

A function $g$ is \textit{supermodular} if for $A \subseteq B \subseteq V$ and $e \in V \setminus B$:
\[ g(A \cup \{e\}) - g(A) \le g(B \cup \{e\}) - g(B). \]
In this work, we focus on functions that are \textit{monotone} ($f(A) \le f(B)$ for $A \subseteq B$) and \textit{grounded} ($f(\varnothing)=0, g(\varnothing)=0$).

As in (\ref{Eqn:SupermodularCurvature}), we define the \textit{curvature} of a monotone supermodular function $g$ as $\gamma = 1 - \min_{g(T)\geq g(S)} \frac{g(T)}{g(T | S)}$. The value of $\gamma$ lies in $[0, 1]$, if $\gamma =0$, the function is modular. In work~\cite{bai2018greed}, curvature is defined as $\gamma' = 1 - \min_{v, S} \frac{g(v)}{g(v | S)}$. Our definition is more general since $\gamma \leq \gamma'$, allowing for a broader class of supermodular functions to be analyzed. As an example, consider the function $g(S) = |S|^2$. Based on our definition:
\begin{align*}
  \min_{g(T) \geq g(S)} \frac{g(T)}{g(T | S)}
  &\geq \min_{g(T) \geq g(S)} \frac{|T|^2}{( (|T|+|S|)^2 - |S|^2 )}  \\
  &\geq\min_{n \geq 1} \frac{n^2}{((2n)^2 - (n)^2)} \geq 1/3.
\end{align*}
As a result, the curvature $\gamma \leq 1 - \frac{1}{3} = \frac{2}{3}$.
On the other hand, using the definition from~\cite{bai2018greed}, $\gamma' = 1 - \min_{v, S} \frac{g(v)}{g(v | S)} = 1 - \min_{v, S} \frac{1}{2|S| + 1} = 1$. 

We will further strengthen our analysis under the assumption that the curvature of the submodular objective 
$f$ is bounded. We use the standard definition of submodular curvature:
\begin{eqnarray}\label{Eqn:SubmodularCurvature}
c \;=\; 1 - \min_{v,S}\frac{f(v\mid S)}{f(v)}\in[0,1).
\end{eqnarray}

\section{Methods and Analysis}

The greedy algorithm is an intuitive and conceptually simple approach to subset selection problems. The selection of the locally optimal element at each step is easy to communicate to non-experts and has proved remarkably effective in practice. Although more sophisticated optimization techniques or specialized heuristics could be applied to submodular maximization under supermodular costs, we deliberately focus on the greedy paradigm. We demonstrate that a refined variant of the greedy rule suffices to obtain strong bicriteria approximation guarantees. Our results indicate that even in the presence of highly non-linear constraints, greedy methods remain both theoretically robust and computationally efficient.


\subsection{Greedy Selection based on Ratio-Marginal}
We use a greedy rule, which at each iteration picks the element with the largest ratio of marginal objective gain to marginal cost, until the constraint is reached or exceeded (Algorithm~\ref{alg:greedy}). 
The returned set may slightly violate the budget; we give bicriteria bounds that trade approximation for constraint violation.

\begin{algorithm}[h]
\caption{Greedy ratio algorithm}
\label{alg:greedy}
\begin{algorithmic}[1]
\Statex \textbf{Input:} ground set $V$, functions $f,g$, and budget $\theta$ 
\State $S_0 \leftarrow \varnothing, i\gets 0$
\While{$\exists v \notin S_i$ and $g(S_i) < \theta$}
    \State $v \gets \arg \max_{u \in V\setminus S_i} \dfrac{f(u \mid S_i)}{g(u \mid S_i)}$
    \State $S_{i+1} \leftarrow S_{i} \cup \{v\},\,\,\, i\gets i+1$
\EndWhile
\State \Return $S$
\end{algorithmic}
\end{algorithm}

\subsection{Main Guarantees}
We first present bicriteria guarantees that depend on the curvature of $g$ that analyzes the bounds when the greedy algorithm overflows the constraint for the first time. The steps leading to this analysis also reveal the approximation obtained from earlier steps.
We will then extend the results to the case when the curvature of $f$ is also known. Then we explore how the approximation improves as we continue to grow the set beyond the constraint.

\begin{restatable}[Bicriteria Approximation with Curvature]{theorem}{main-thm}\label{thm:main-thm}
For any monotone submodular $f$ and
supermodular $g$ with curvature $\gamma < 1$. If the greedy algorithm returns a set $S_k$, then
\begin{eqnarray}
g(S_k) &\leq & \frac{2 - \gamma}{1 - \gamma} \cdot \theta, \quad {\rm and }\\
f(S_k) & \geq & \left(1 - e^{-(1 - \gamma)} \right) f(S^*),
\end{eqnarray}
where $S^*$ is an optimal solution to $\max f(S)$ subject to $g(S) \leq \theta$.
\end{restatable}

\begin{proof}
Let $S_i$ be the current solution produced by the greedy algorithm, and let $S^* = \{v_1^*, v_2^*, \dots\}$ be the optimal solution.
Define a sequence of sets $A_j := S_i \cup \{v_1^*, \dots, v_j^*\} \quad \text{for } j = 0, 1, \dots, k.$
By the greedy selection at iteration $i+1$, we have
\begin{equation}
\frac{f(v_{i+1} \mid S_i)}{g(v_{i+1} \mid S_i)} \geq \max_j \frac{f(v_j^* \mid S_i)}{g(v_j^* \mid S_i)},
\end{equation}
where $\{v_j^*\}$ are elements of the optimal solution $S^*$.
By inequality
$
\max_j \frac{a_j}{b_j} \geq \frac{\sum_j a_j}{\sum_j b_j}
$
for positive sequences $\{a_j\}$ and $\{b_j\}$, we have
\begin{equation}
\max_j \frac{f(v_j^* \mid S_i)}{g(v_j^* \mid S_i)} \geq \frac{\sum_j f(v_j^* \mid S_i)}{\sum_j g(v_j^* \mid S_i)}.
\end{equation}

By submodularity of $f$ and supermodularity of $g$, along with the definitions of the intermediate sets $A_j = S_i \cup \{v_1^*, \ldots, v_j^*\}$, we have
\begin{equation}
\sum_j f(v_j^* \mid S_i) \geq \sum_j f(v_j^* \mid A_{j-1}) = \sum \left(f(A_j) - f(A_{j-1})\right) = f(S^* \mid S_i)
\end{equation}
and
\begin{equation}
\sum_j g(v_j^* \mid S_i) \leq \sum_j g(v_j^* \mid A_{j-1}) = \sum \left(g(A_j) - g(A_{j-1})\right) = g(S^* \mid S_i)
\end{equation}
Combining these inequalities yields
\begin{equation}
\frac{f(v_{i+1} \mid S_i)}{g(v_{i+1} \mid S_i)} \geq \frac{f(S^* \mid S_i)}{g(S^* \mid S_i)}.
\end{equation}
Multiplying both sides by $g(v_{i+1} \mid S_i)$:
\begin{equation}
f(v_{i+1} \mid S_i) \geq \frac{f(S^* \mid S_i)}{g(S^* \mid S_i)} \cdot g(v_{i+1} \mid S_i).
\end{equation}
\begin{equation}\label{eqn:first-recur}
\implies f_{i+1} - f_i \geq \frac{f^{\circledast} - f_i}{g(S^* \mid S_i)} \cdot (g_{i+1} - g_i),
\end{equation}
where $f_i := f(S_i)$, and $f(S^* \cup S_i) := f^{\circledast} \geq  f(S^*)$, and $g_i := g(S_i)$.
Using the curvature bound:
\begin{equation}\label{eqn:using_curvature}
g(S^* \mid S_i) \leq \frac{g(S^*)}{1 - \gamma},
\quad \Rightarrow \quad \frac{1}{g(S^* \mid S_i)} \geq \frac{1 - \gamma}{g(S^*)}.    
\end{equation}
Substituting this in Inequality~\ref{eqn:first-recur},
\begin{equation}
f_{i+1} - f_i \geq (f^{\circledast} - f_i) \cdot \frac{1 - \gamma}{g(S^*)} \cdot (g_{i+1} - g_i).
\end{equation}

Let $G_i = \frac{g_i}{g(S^*)}$. Then the recurrence becomes:
\begin{equation}\label{eqn:first_recur}
f_{i+1} - f_i \geq (f^{\circledast} - f_i)(1-\gamma)(G_{i+1} - G_i).
\end{equation}

We now show by induction $f_i \geq \left(1 - e^{-(1-\gamma)G_i} \right) f^{\circledast}.$
The base case is straight-forward to verify with $f_0 = 0$, $G_0 = 0$. Suppose $f_i \geq (1 - e^{-(1-\gamma)G_i}) f^{\circledast}$. Then:
\begin{align*}
f_{i+1} &\geq f_i + (f^{\circledast} - f_i)((1-\gamma)G_{i+1} - (1-\gamma)G_i) \\
&\geq \left(1 - e^{-(1-\gamma)G_i} \right) f^{\circledast} + \left(e^{-(1-\gamma)G_i} f^{\circledast} \right)((1-\gamma)G_{i+1} - (1-\gamma)G_i) \\
&= f^{\circledast} \left(1 - e^{-(1-\gamma)G_{i+1}} \right).
\end{align*}
Therefore,
\begin{equation}
f_k \geq (1 - e^{-\tilde{g}_k}) f^{\circledast} \geq (1 - e^{-(1-\gamma)G_k}) f^{*}.
\end{equation}

Now suppose the algorithm halts the first time $g(S_k) \geq \theta$, (i.e., $g(S_{k-1}) < \theta$). Note that the above bound applies only for $g(S_{k-1}) \leq g(S^*)$ due to the use of our curvature definition. However, if $g(S_{k'}) > g(S^*)$ for some $k' \leq k-1$, then we already have
\begin{align}
f_{k'+1} &\geq  (1 - e^{-(1-\gamma)G_{k'+1}})f^* = \left(1 - e^{-(1-\gamma)f^*  \frac{g(S_{k'+1})}{g(S^*)}}\right)f^* \\
&\geq  (1-e^{-(1-\gamma)})f^*\,.  \label{eqn:early_terminate}
\end{align}
That is, if we had stopped earlier, we would have the desired approximation already. Due to the monotonicity of $f(\cdot)$, the approximation guarantee is preserved.
Now, for the violation of the constraint, 
suppose $g(v_k) \geq g(S_{k-1})$. Then, by curvature,
\begin{align*}
\frac{g(v_k)}{g(v_{k} \mid S_{k-1})} \geq 1-\gamma 
\implies \frac{g(v_k)}{g(S_{k}) - g(S_{k-1})} &\geq 1-\gamma \\
\implies g(S_{k}) \leq \frac{g(v_{k})}{1-\gamma} + g(S_{k-1}) 
\implies g(S_{k}) &\leq \frac{\theta}{1-\gamma} + \theta 
\\
&= \frac{2-\gamma}{1-\gamma}\theta.
\end{align*}
Here, we use $g(v_k) \leq \theta$ -- we can remove any element from the given set that by itself would have violated the constraint. The same can be derived if $g(S_k) \geq g(v_k)$.

Therefore, $f(S_k) \geq \left(1 - e^{-(1-\gamma)}\right)f(S^*)$ and $g(S_k) \leq \frac{2-\gamma}{1-\gamma}\theta$.
\end{proof}

\textbf{Remark:} One may expect to get a $\left(1 - e^{-(1-\gamma)}\right)$ approximation without violating the constraint. However, the scenario leading to the constraint violation is that $g(S_{k-1}) = \theta - \epsilon$ for very small $\epsilon$. And the next element adds the maximum marginal increase allowable by the curvature, leading to $g(S_k) = \frac{2-\gamma}{1-\gamma}\theta - \epsilon$. In practice, we can obtain a better guarantee once we know the value of $g(S_i)$ for any step $i$ of the greedy algorithm.

\begin{corollary}\label{cor:before_overflow}
    If $g(S_i) = \beta \theta$, then the Greedy at step $i$ results in a $\left(1-e^{-\beta(1-\gamma)}, \beta \right)$ approximation for $\beta < 1$.
\end{corollary}

The proof follows from the fact that $(1-\gamma)G_i = (1-\gamma)\frac{g(S_i)}{g(S^*)} \geq (1-\gamma)\frac{\beta \theta}{\theta} = \beta(1-\gamma)$. 

%
%

    
    

We can now prove Theorem \ref{thmBoundary}, which suggests that either we get the desired approximation before the algorithm terminates, or we get an improved factor at termination at the cost of overflow.
\begin{proof}(of Theorem \ref{thmBoundary})
If $g(S_{k'}) \geq g(S^*)$ for any $k' < k$, then we have the desired approximation as noted in Equation~\ref{eqn:early_terminate}. If not, then $S_k$ is the first time $g(S_k) > g(S^*)$, which means $g(S_{k-1}) \leq g(S^*)$, allowing us to follow the derivation of Theorem~\ref{thm:main-thm} to get $\tilde{g}_k = (1-\gamma)\frac{g(S_k)}{g(S^*)} \geq (1-\gamma)\frac{g(S_k)}{\theta} = \beta (1-\gamma)$.
\end{proof}

On the other hand, if we use the strong curvature bound $\gamma'$, then $g(S^*|S_i) \leq \frac{g(S^*)}{1-\gamma'}$ holds for all $S_i$. As a result, we can continue the greedy algorithm beyond the constraint and continue to improve the approximation factor at the cost of higher constraint violations.

\begin{corollary}[Greedy Approximation for Strict Supermodular Curvature]\label{cor:overflow-traditional}
        If $g(S_i) = \beta \theta$, then the greedy algorithm at step $i$ results in a $\left(1-e^{-\beta(1-\gamma')}, \beta \right)$ approximation $\forall \beta$, where $\gamma'$ is using the stricter definition of curvature of $g$.
\end{corollary}


We further prove that there cannot exist a better bound than $(1 - e^{-(1-\gamma)}$ for our Greedy algorithm without violating the constraint.

\begin{figure}
    \includegraphics[width=\columnwidth]{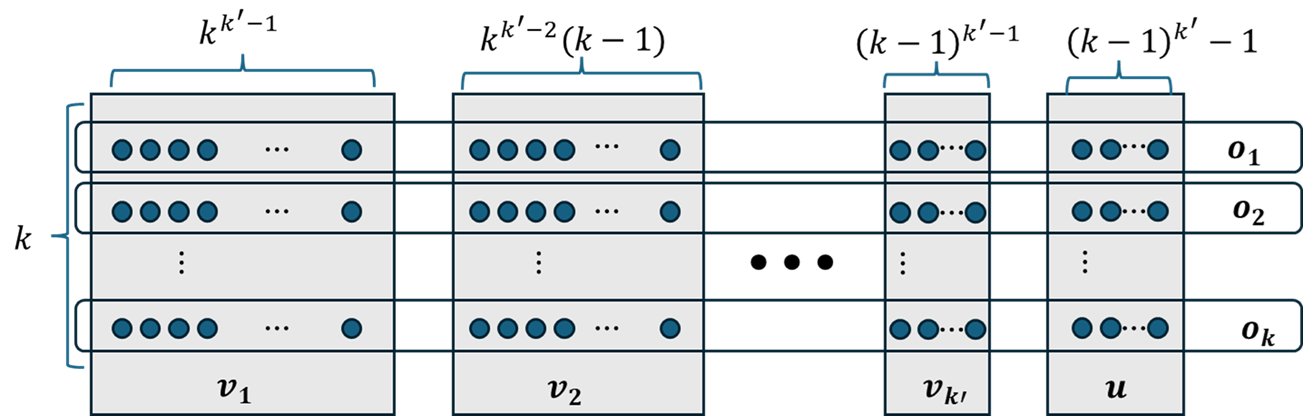}
    \caption{Instance of max cover with supermodular cost.}
    \label{fig:tightness}
\end{figure}

\begin{restatable}[Tightness of Greedy Approximation]{theorem}{thmTightness}\label{thm:tightness}
For any curvature $\gamma \in [0, 1)$, there exists an instance of non-negative monotone submodular maximization with a supermodular knapsack constraint of curvature $\gamma$, such that the approximation ratio of the Greedy algorithm is arbitrarily close to $1 - e^{-(1-\gamma)}$.
\end{restatable}

\begin{proof}
Consider the max-cover where we wish to cover the maximum number of elements by selecting subsets from a given collection while keeping the cost under $k$. For a given $d \in [0, 1)$, we define a positive integer $k'$ obtained from some function of $k$ and $\gamma$. We will later find an appropriate choice of $k'$.

We create an instance of this problem as shown in Figure~\ref{fig:tightness}. It consists of subcollection $V$ with subsets $\{v_1, \dots, v_{k'}\}$, subcollection $O$ with subsets $\{o_1, \dots, o_{k}\}$, and a subset $u$.

\paragraph{The Submodular Coverage Function $f$:}
Let $f(S) = |\bigcup_{s \in S} s|$ be the set coverage function. Based on our construction:
\[
f(v_i) = k\cdot k^{k'-i} \cdot (k-1)^{i-1} = k^{k'-(i-1)} \cdot (k-1)^{i-1},
\]
\[
f(u) = (k-1)^{k'} - 1\,,
\]

\begin{align*}
\text{and } f(o_i) &= \sum_{i=1}^{k'} k^{k'-i} \cdot (k-1)^{i-1} + (k-1)^{k'}-1 \\
&= (k - (k-1))\sum_{i=1}^{k'} k^{k'-i} \cdot (k-1)^{i-1} + (k-1)^{k'}-1 \\
&= k^{k'} - (k-1)^{k'} + (k-1)^{k'} - 1 =  k^{k'} - 1\,.
\end{align*}

\paragraph{The Supermodular Cost Function $g$, and choice of $k'$:}
We define a cost function $g: 2^{\mathcal{O} \cup \mathcal{G}} \to \mathbb{R}_{\geq 0}$ with a budget $B = k$. The individual costs are:
\begin{align*}
g(o_j) = 1 \quad \forall o_j, \,\,
g(v_i) = 1 \quad \forall v_i, \,\,
g(u) = \infty
\end{align*}
Due to $g(u)=\infty$, $u$ can never be selected by an algorithm.  The function $g$ is modular for all subsets $S \subseteq \mathcal{G}$ (i.e., $g(S) = |S|$). However, we introduce a supermodular ``jump'' when mixing sets. For the specific set $V_{k'} = \{v_1, \dots, v_{k'}\}$, and any $o_j \in \mathcal{O}$, we define:
\[
g(o_j \cup V_{k'}) = k + \epsilon
\]
where $\epsilon > 0$ is a small constant. As a result, the curvature of $g$ is
\begin{align}
    \gamma &= 1 - \frac{g(V_{k'})}{g(V_{k'}|o_j)} = 1 - \frac{k'}{k+\epsilon - 1} \nonumber \\
    \implies k' &= (k+\epsilon)(1-\gamma)\,.
\end{align}
Note that for any $\gamma$ and large enough integer $k$, we can choose a small epsilon so that $k'$ is an integer.

\paragraph{Execution of the Greedy Algorithm}

The Greedy algorithm iteratively selects the element $s$ maximizing the ratio $\rho(s) = \frac{\Delta f(s)}{\Delta g(s)}$.
We will show that the greedy algorithm will pick $v_1, v_2, \dots, v_{k'}$. For the first pick, 

\[
f(v_1) = k^{k'}, g(o_i) = k^{k'} - 1\,.
\]
So, $v_1$ is selected.
At step $i>1$, assume the greedy set already prepared is $V_{i-1} = \{v_1, v_2, \dots, v_{i-1}\}$. 

\noindent\textbf{Candidate $v_i$:} The marginal gain is
    $f(v_i | V_{i-1}) = k^{k'-(i-1)} \cdot (k-1)^{i-1}$ and at marginal cost 1.

\noindent \textbf{Candidate $o_j$:} 
    The best that $o_j$ additionally covers is the remaining uncovered elements in $o_j$. This is
    \begin{align*}
    &f(o_j|V_{i-1}) =\sum_{l=i}^{k'} k^{k'-l}(k-1)^{l-1} + (k-1)^{k'-1} - 1 \\
    &= (k-1)^{i-1}\sum_{l=1}^{k'-(i-1)} k^{k'-(i-1) - l}(k-1)^{l-1} + (k-1)^{k'-1} - 1 \\
    &=(k-1)^{i-1} \left( k^{k'-(i-1)} - (k-1)^{k' - (i-1)}\right) + (k-1)^{k'-1} - 1 \\
    &= k^{k' - (i-1)} (k-1)^{i-1} - 1\,.
    \end{align*}

Comparing the two, it follows that at iteration $i$ of the greedy algorithm, it must select $v_i$.

After $k'$ steps, the total cost incurred is $g(V_{k'}) = k'$. The remaining budget is $k - k'$.

We show that greedy must stop at this point. The only potential candidates are $u$ and one of $o_j$. Greedy can never select $u$ due to inifnite cost. So the only other option is $o_j$ for any $j$.
But $g(V_{k'} \cup o_j) = k + \epsilon$, violating the bound. 

Therefore, the total coverage achieved by greedy is
\begin{align*}
    f(V_{k'}) &= \sum_{i=1}^{k'} k^{k'-(i-1)} \cdot (k-1)^{i-1} \\
    &= k^{k'}\sum_{i=1}^{k'} \left( 1 -1/k\right)^{i-1} = k^{k'+1}\left( 1 - (1-1/k)^{k'}\right)\,.
\end{align*}
And the optimal coverage obtained by $S^* = {o_1, o_2, \dots, o_k}$ is
\begin{equation}
  f(S^*)  = k\cdot f(o_1) = k\cdot \left(k^{k'} - 1\right) = k^{k'+1}(1-1/k^{k'})
\end{equation}
Therefore, the approximation ratio of greedy is
\begin{align*}
    \frac{f(V_{k})}{f(S^*)} 
    &= \frac{k^{k'+1}\left( 1 - (1-1/k)^{k'}\right)}{k^{k'+1}(1-1/k^{k'})} = \frac{ 1 - (1-1/k)^{k'} }{1- 1/k^{k'}}
\end{align*}
Taking the limit as $k$ approaches $\infty$, we get
\begin{equation}
    \frac{f(V_{k})}{f(S^*)} = 1 - e^{-k'/k} = 1 - e^{-(k+\epsilon)(1-\gamma)/k} \approx 1 - e^{-(1-\gamma)} + \epsilon'\,,
\end{equation}
where $\epsilon'$ can be made arbitrarily close to zero with large $k$.
\end{proof}

\subsection{Impact of Submodular Curvature of $f$}

We obtain tighter bounds when the curvature of 
$f$ is also available.

\begin{restatable}{theorem}{thmCurvF}[Bicriteria Approximation with Curvature of $f$]\label{thm:curv-f}
Let $c$ be the curvature of $f$ (see Equation \ref{Eqn:SubmodularCurvature}).
The Greedy algorithm that overshoots the constraint by at most one element leads to
\[
f(S_k) \ge \bigl(1 - \left[1 - (1-c)(1-\gamma)\right]^{\frac{1}{1-c}}\bigr)\,f(S^*),
\]
\end{restatable}

\begin{proof}[Proof idea]
To derive the guarantee, we define $f_i = f(S_i)$ and $g_i = g(S_i)$ as the objective and cost at iteration $i$, and let $f^* = f(S^*)$ and $g^* = g(S^*)$ be the values for an optimal solution $S^*$. We further partition the Greedy set $S_i$ into elements belonging to the optimal set, $\hat{g}_i = g(S_i \cap S^*)$, and other elements, $\bar{f}_i = f(S_i \setminus S^*)$ and $\bar{g}_i = g(S_i \setminus S^*)$.

By analyzing the marginal gains $\Delta f_{i+1} = f_{i+1} - f_i$ and $\Delta g_{i+1} = g_{i+1} - g_i$, and using the Greedy property, we can show that $\bar{f}_i \ge \frac{\Delta f_{i+1}}{\Delta g_{i+1}} \bar{g}_i$. Using the submodular curvature $c$, we lower-bound the marginal progress $f(S^* \mid S_i) \ge f^* - f_i + (1-c)\bar{f}_i$, and similarly use supermodular curvature $\gamma$ to bound the marginal cost $g(S^* \mid S_i)$. Combining these with the Greedy choice property, we obtain the discrete recurrence:
\[
  (g^* - (1-c)(1-\gamma) g_i)\,\Delta f_{i+1} \ge (f^* - f_i)(1-\gamma)\,\Delta g_{i+1},
\]
We then transform this recurrence using Bernoulli's inequality into one that can be solved by telescoping. 
See appendix for  details.
\end{proof}

\subsubsection{Continuing the Overflow}
In Corollary~\ref{cor:overflow-traditional}, we noted that with the traditional definition of supermodular curvature, we could continue the idea of proof of Theorem~\ref{thm:main-thm} beyond the constraint and continue to improve the approximation to get $\left(1 - e^{-\beta(1-\gamma)}, \beta\right)$. However, the same idea does not directly apply to our setting as we use the fact $g(S^*) \geq g(S_i)$. Therefore, we need to separately derive a recurrence relation that is valid when $g(S^*) < g(S_i)$. The following theorem bounds the approximation factor when we exceed the constraint by an arbitrary factor $\beta^+ > 1$.

\begin{restatable}[Bound Beyond Overflow]{theorem}{thmBoundBeyond}\label{thm:general}
    Suppose the Greedy algorithm is run for $K$
    and $g(S_K) = \beta^+ \theta$, then
    \[
    f(S_K) \geq 1 - \left(\frac{\beta^+ + \frac{1-\gamma}{\gamma}}{1 + \frac{1-\gamma}{\gamma}}\right)^{-\frac{1-\gamma}{\gamma}}(1 - (1-c)(1-\gamma))^{\frac{1}{1-c}}
    \]
\end{restatable}

\begin{proof}[Proof Idea]
 We obtain a 
 new marginal cost bound for $i \geq k$:
 $$g(S^*|S_i) \leq \frac{\gamma}{1-\gamma}g_i + g^*.$$
 This allows us to derive a new recurrence relation for the objective gain $\Delta f_{i+1}$ relative to the cost increase $\Delta g_{i+1}$. By defining normalized objective and cost variables, $F_i$ and $G_i$, we transform the recurrence into a form where we can apply a generalized Bernoulli Inequality -- specifically the form $1 - p'x \leq (1+x)^{-p'}$ for $p' > 0$. This transformation converts the additive growth of the objective into a multiplicative bound. By telescoping this relation from the initial overflow step $k$ through to an arbitrary final step $K$, we derive a bound that depends on both the initial approximation at step $k$ and the overflow factor $\beta^+ = g_K / g^*$. 
 The detailed proof is available in Appendix.
\end{proof}

Theorem~\ref{thm:general} and Corollary~\ref{cor:before_overflow} combined provide the approximation factor for any step of Greedy (before or after the constraint) as shown in Figure~\ref{fig:teaser}.

\subsection{The Dual Problem} \label{Sec:Dual}

Using the Greedy algorithm guarantees, in conjunction with a binary search technique, we can obtain guarantees for the dual problem as well, where we wish to minimize the supermodular function $g$ subject to a lower bound constraint on the submodular function $f$.

\begin{algorithm}[h]
\caption{Dual algorithm via binary search}
\begin{algorithmic}[1]
\Statex \textbf{Input:} target value $\tau$, precision $\varepsilon$, primal algorithm $\mathcal{A}$
\State $L \leftarrow 0, R \leftarrow g(V)$
\While{$R - L > \varepsilon$}
    \State $B \leftarrow (L + R)/2$
    \State $S \leftarrow \mathcal{A}(B)$
    \If{$f(S) \geq \alpha \tau$}
        \State $R \leftarrow B, S_{best} \leftarrow S$
    \Else
        \State $L \leftarrow B$
    \EndIf
\EndWhile
\State \Return $S_{best}$
\end{algorithmic}
\end{algorithm}

\begin{restatable}{theorem}{thmDual}[Dual approximation]\label{thmDual}
Let $f:2^V\to\mathbb{R}_+$ be a nondecreasing submodular function and
$g:2^V\to\mathbb{R}_+$ a nondecreasing supermodular function.
Suppose we have an algorithm $\mathcal{A}(B)$ which, for any budget $B$, returns a set $S=\mathcal{A}(B)$ satisfying the primal $(\alpha,\beta)$-guarantee (namely,
$f(S) \ge \alpha\cdot\mathrm{OPT}_P(B), g(S) \le \beta B,$
)
where $\mathrm{OPT}_P(B)=\max_{T:\, g(T)\le B} f(T)$.
Let the dual optimum be
\[
B^* \;=\; \min_{T:\, f(T)\ge\tau} g(T).
\]
Then, by performing binary search on $B$ with additive precision $\varepsilon>0$, one can find a set $S$ such that
\[
  f(S) \ge \alpha\tau, \qquad g(S) \le \beta\Big(1+\frac{\varepsilon}{B^*}\Big)B^*,
\]
i.e., binary search with $\mathcal{A}$ yields an algorithmic $(\beta(1+\varepsilon/B^*),\alpha)$-approximation for the dual problem.
\end{restatable}

\begin{proof}[Proof idea]
We approach the dual problem by reducing it to a series of primal problems solved via binary search. Given a target objective $\tau$ and an $(\alpha, \beta)$-approximation algorithm $\mathcal{A}$ for the primal problem, we search for the minimum budget $b$ such that $\mathcal{A}(b)$ satisfies $f(S) \ge \alpha\tau$. By the primal guarantee, if the algorithm fails to meet $\alpha\tau$ at budget $b-\varepsilon$, then the true optimal primal value at that budget is less than $\tau$; this implies that the optimal dual cost $B^*$ must exceed $b-\varepsilon$, effectively bounding the budget $b$ within an additive $\varepsilon$ of the dual optimum. Finally, applying the primal cost guarantee $\beta$ to the budget $b$ ensures that the resulting set $S_1$ satisfies the dual constraint $f(S_1) \ge \alpha\tau$ with a total cost $g(S_1)$ within a factor of $\beta(1 + \varepsilon/B^*)$ of the optimal dual cost $B^*$. The complete proof is available in the Appendix.
\end{proof}

\begin{figure*}[!htbp]
    \centering
    \begin{subfigure}[b]{0.33\textwidth}
        \centering
        \includegraphics[width=\linewidth]{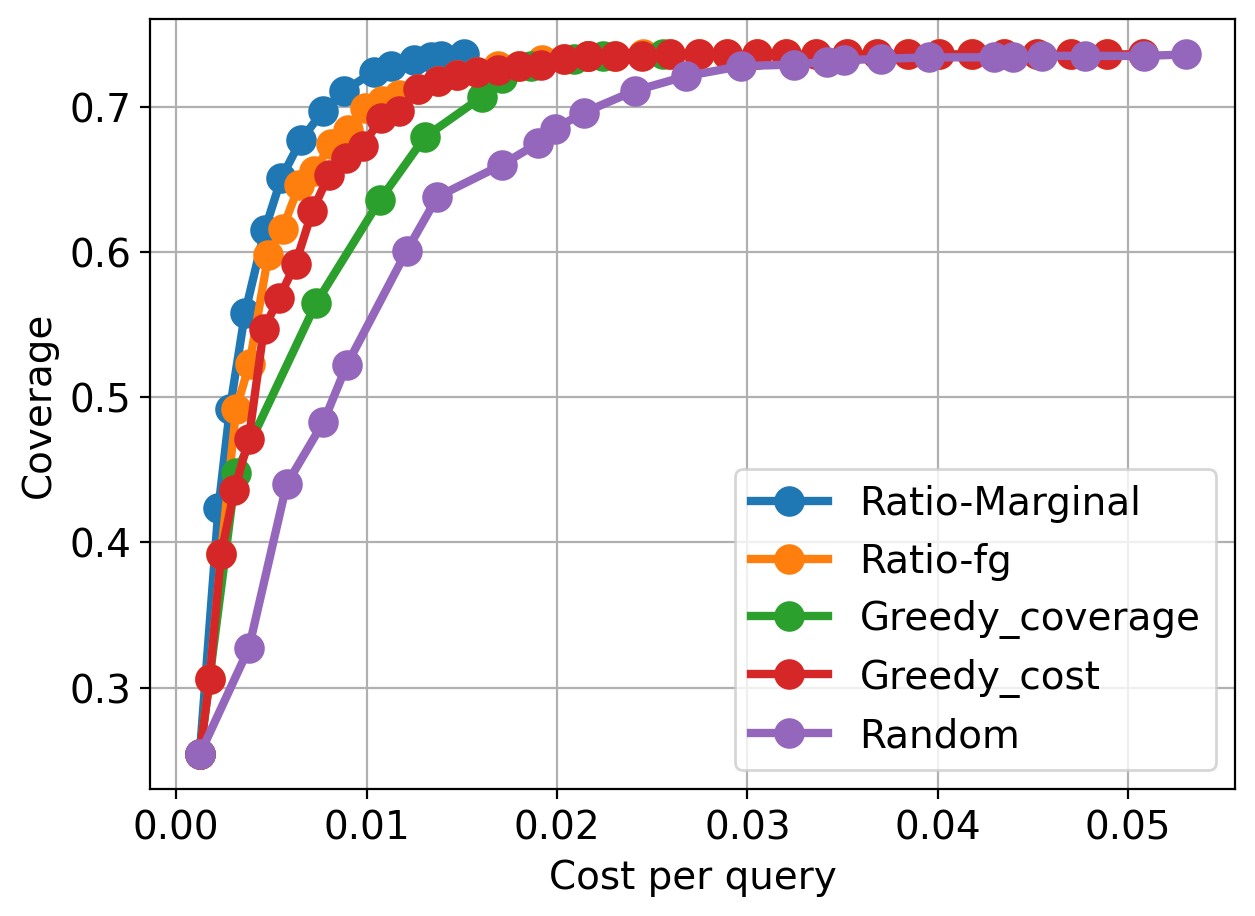}
        \caption{One round}
        \label{fig:sub1}
    \end{subfigure}
    \hfill
    \begin{subfigure}[b]{0.33\textwidth}
        \centering
        \includegraphics[width=\linewidth]{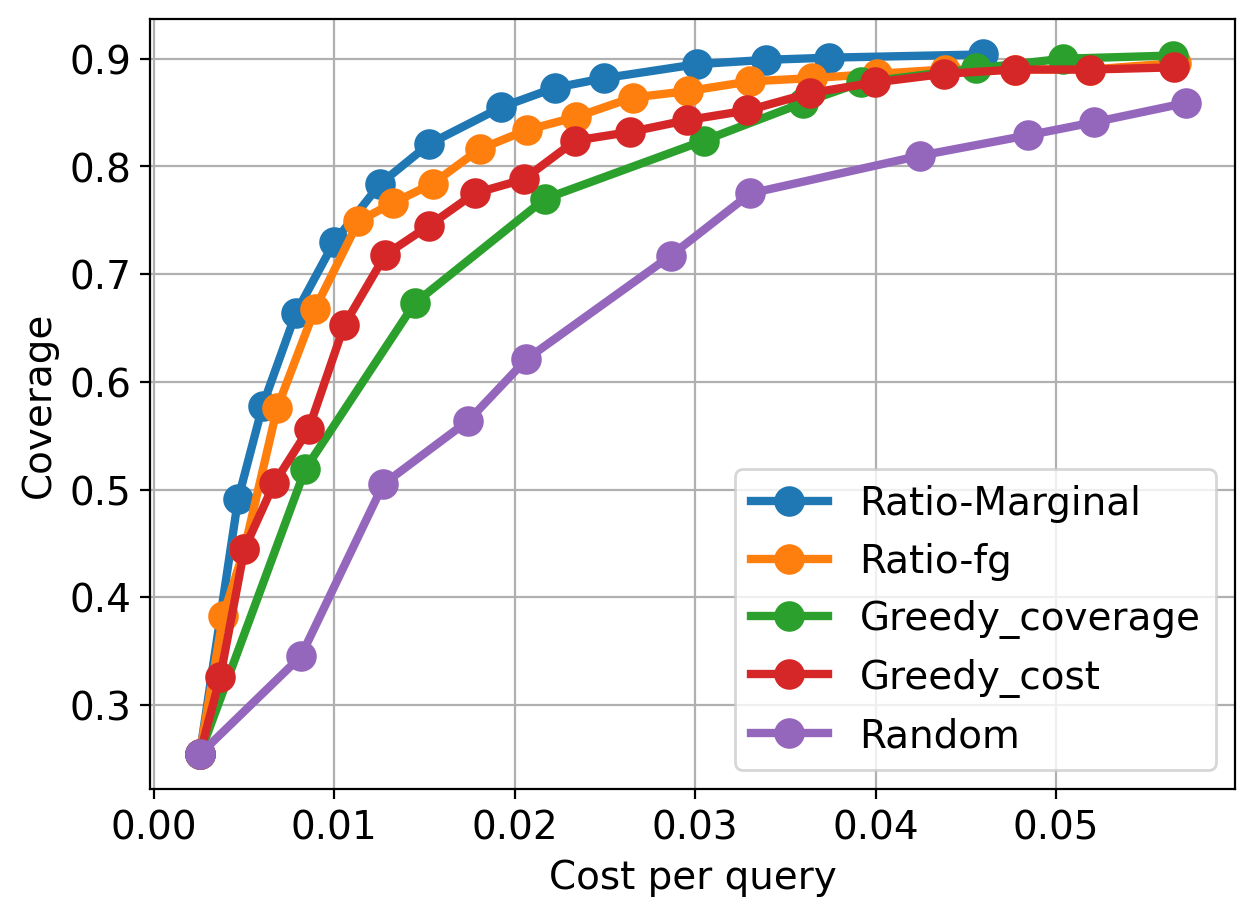}
        \caption{Two rounds}
        \label{fig:sub1}
    \end{subfigure}
    \hfill
    \begin{subfigure}[b]{0.33\textwidth}
        \centering
        \includegraphics[width=\linewidth]{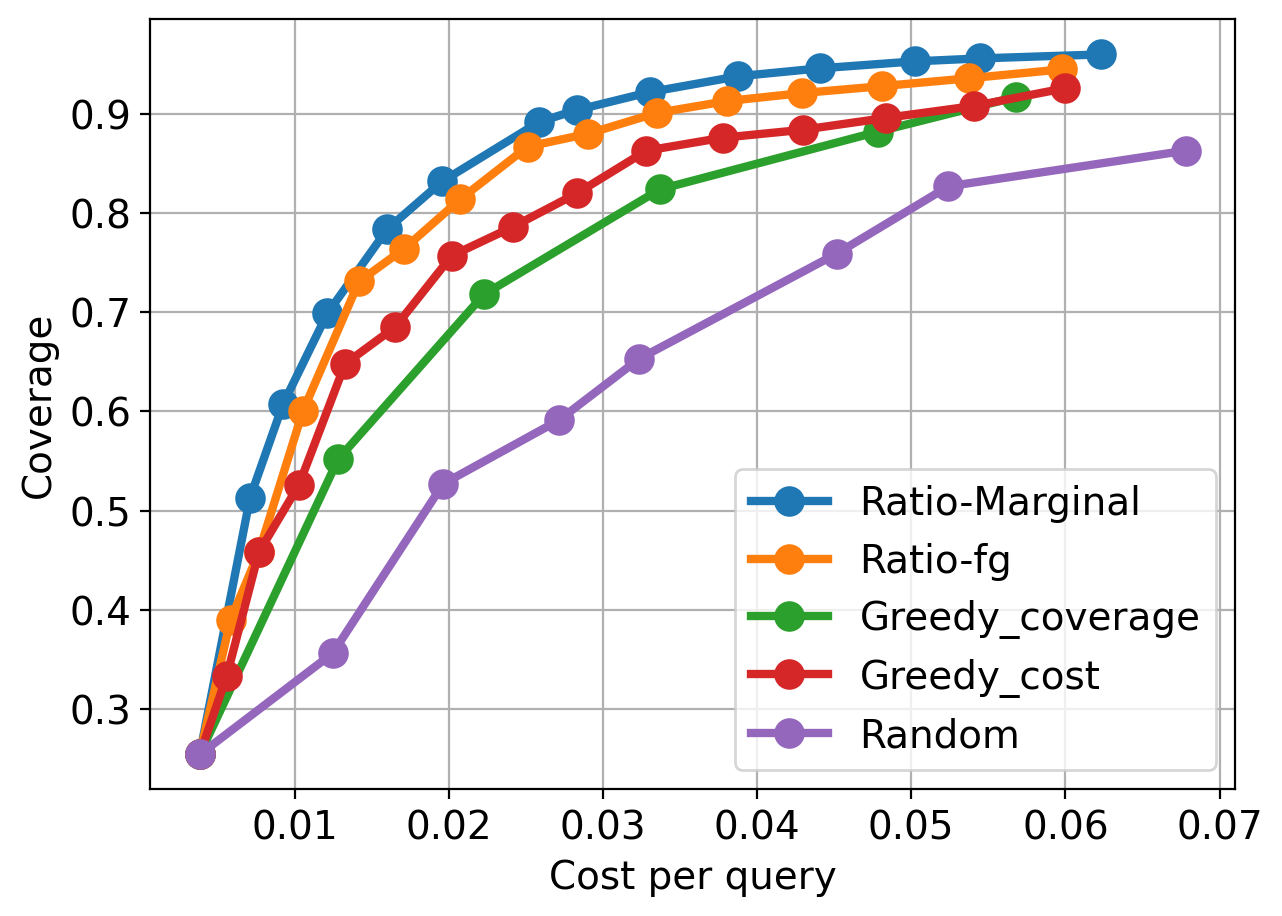}
        \caption{Three rounds}
        \label{fig:sub1}
    \end{subfigure}
    \caption{Comparison with 100 agents and 1000 questions with the ``Global View'' model.}
    \label{fig:global}
\end{figure*}

\begin{figure*}[!htbp]
    \centering
    \begin{subfigure}[b]{0.33\textwidth}
        \centering
        \includegraphics[width=\linewidth]{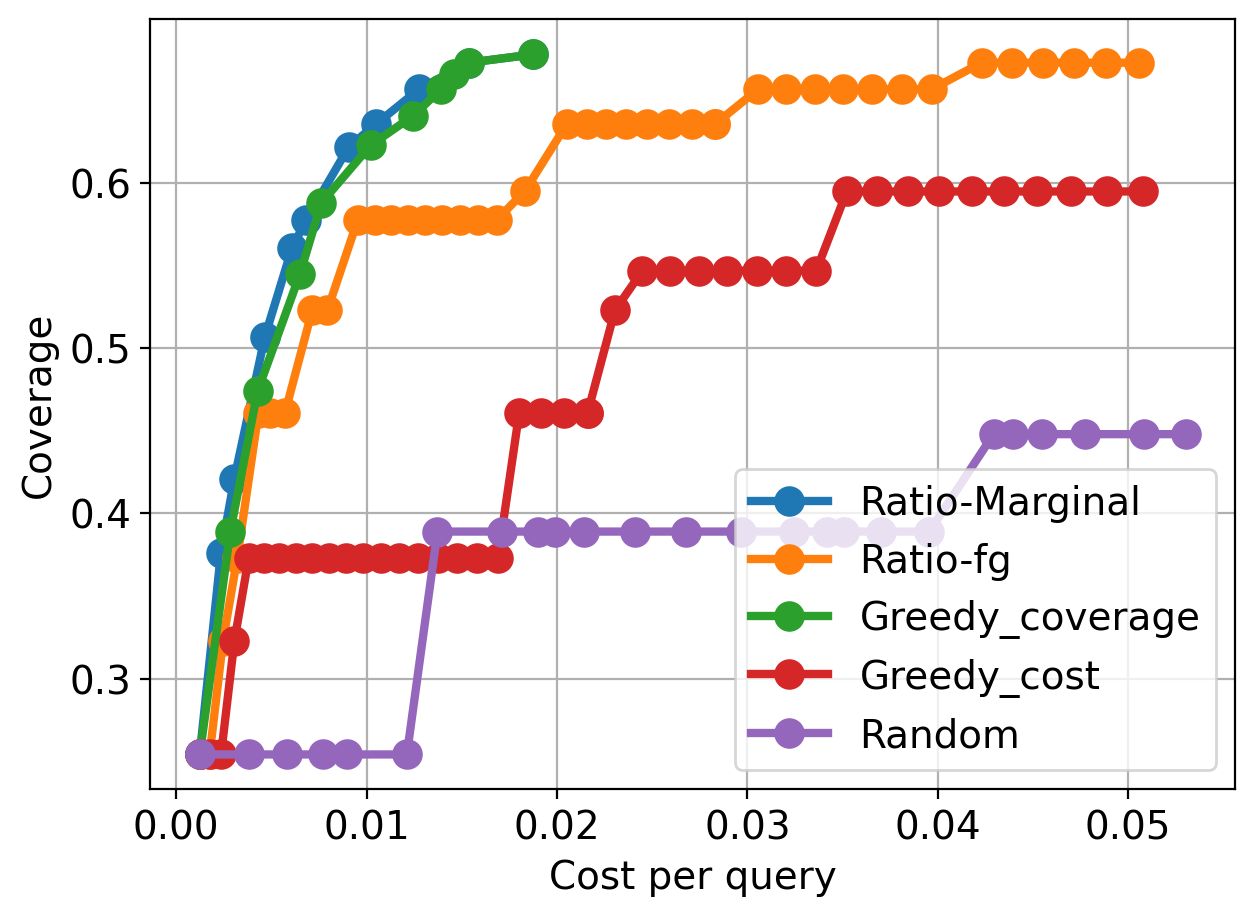}
        \caption{One round}
        \label{fig:sub1}
    \end{subfigure}
    \hfill
    \begin{subfigure}[b]{0.33\textwidth}
        \centering
        \includegraphics[width=\linewidth]{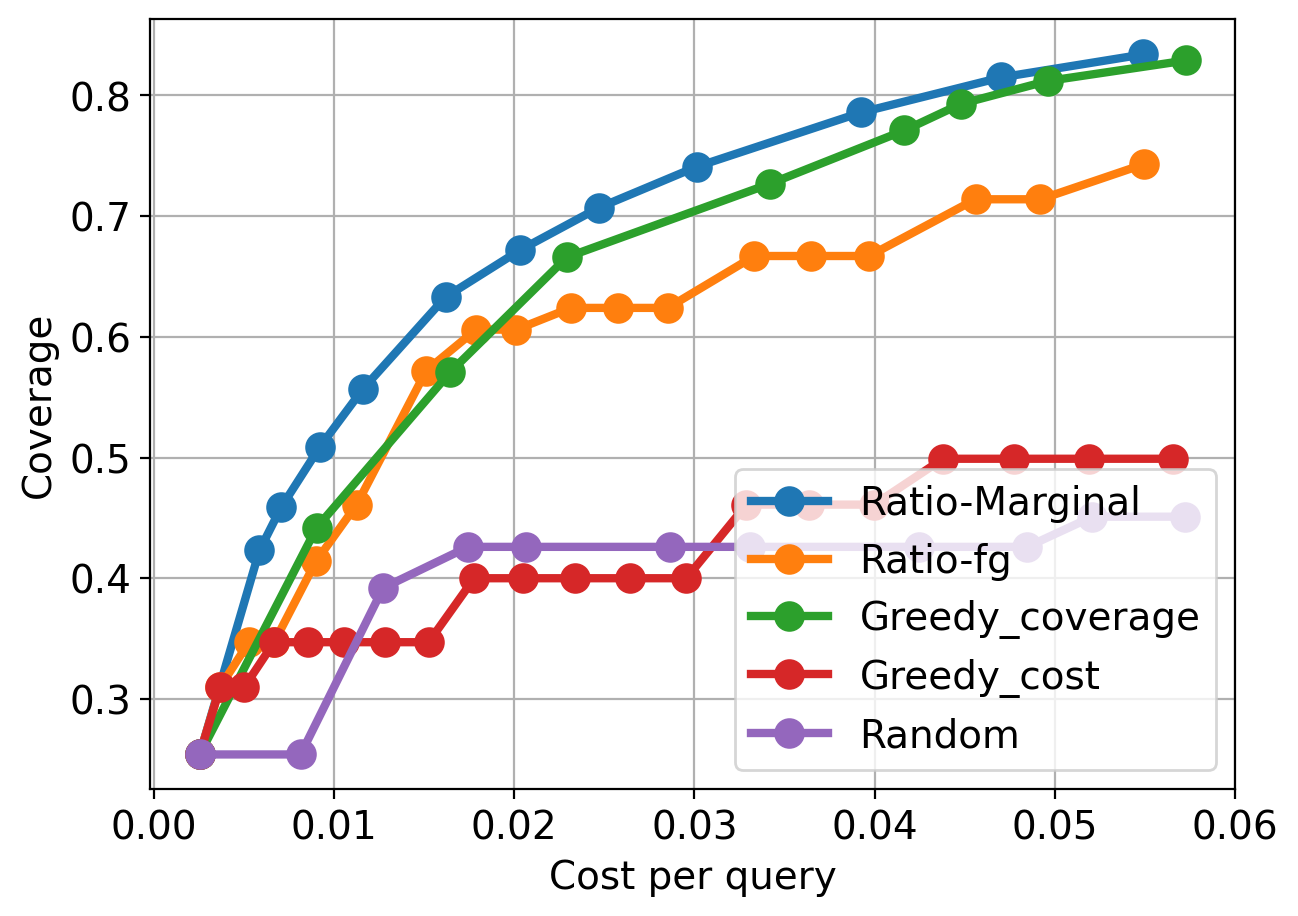}
        \caption{Two rounds}
        \label{fig:sub1}
    \end{subfigure}
    \hfill
    \begin{subfigure}[b]{0.33\textwidth}
        \centering
        \includegraphics[width=\linewidth]{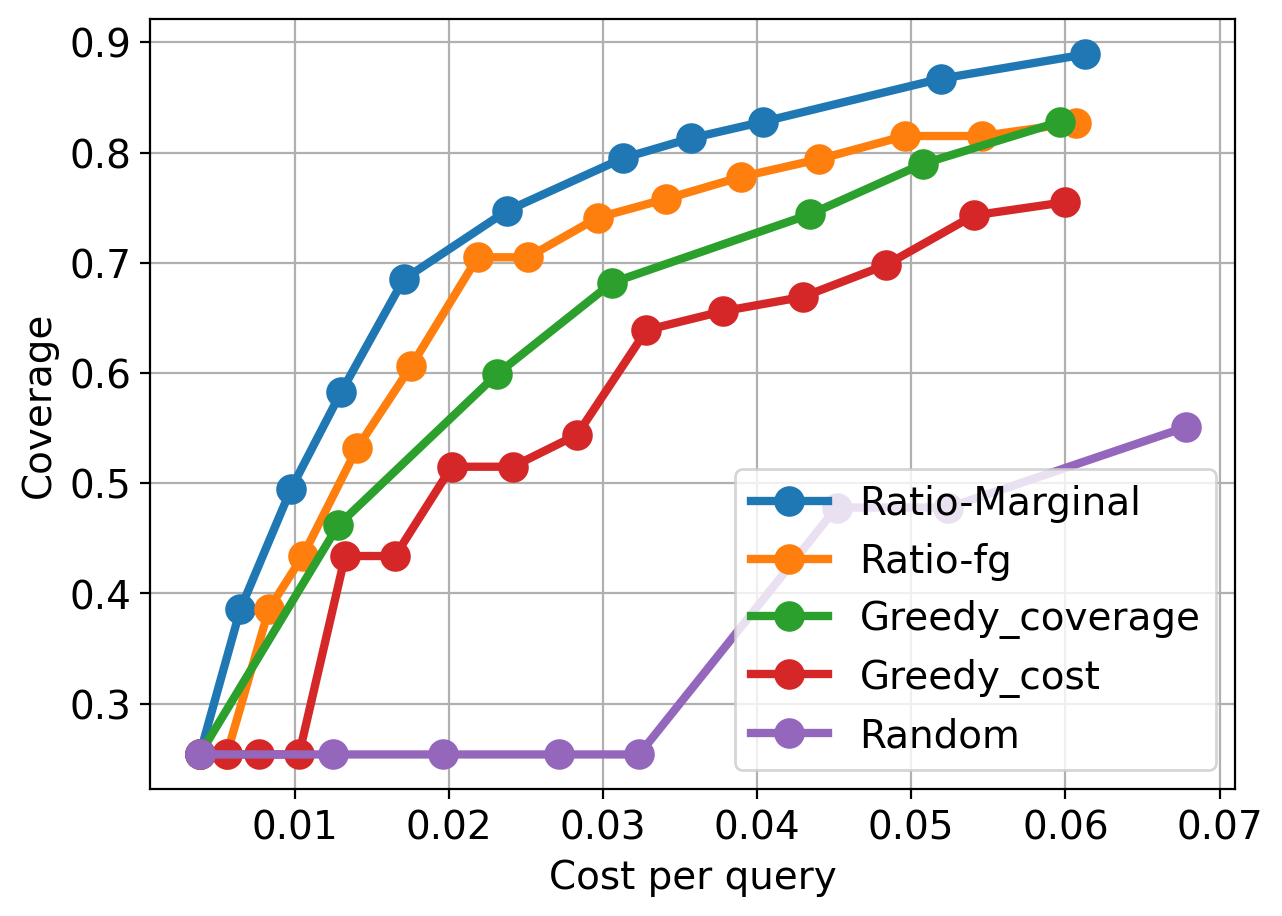}
        \caption{Three rounds}
        \label{fig:sub1}
    \end{subfigure}
    
    \caption{Comparison with 100 agents and 1000 questions with the ``Local View'' model.}
    \label{fig:local}
\end{figure*}

\section{Experiments}
To compare our Greedy algorithm with other greedy heuristics, we ran simulations of debating agents~\cite{du2023improving}. 

\subsection{Simulation Setup}

We implement a synthetic debate simulator where agents debate over multiple rounds to improve their responses~\cite{du2023improving}. Particularly, there is an agent-0 who would debate other agents hoping to get the right answer at the end of multiple rounds. At each round, agents consume initial inputs (system prompts and user prompts), and outputs of the previous round from all the models to generate outputs. Consider a benchmark of $T$ questions, using which we can evaluate the performance of agents before they are deployed. However, as more agents are selected the cost of the workflow increases due to increased inputs and outputs. The task is to pick a set of agents so that, eventually, agent-0 can correctly answer as many questions as possible in this benchmark while keeping the cost low. Below, we summarize the simulator mechanics and the parameters. All our code is publicly available~\footnote{\url{https://github.com/scc-usc/sub-sup-KDD2026}}

\paragraph{Agent generation.} We sample $m$ agents and $T$ binary-labeled questions. For each agent $i$, we draw a base accuracy $p_i$
and mark $\lfloor p_iT\rfloor$ randomly chosen questions as answered correctly by agent $i$ before the first round of debate. For each agent, we pick a random cost for input and output tokens. We make these costs partially correlated with the correctness to reflect a realistic scenario where better LLMs generally cost more. And similar to typical LLMs, we make output tokens more expensive than input tokens. 

\paragraph{Debate dynamics.} We implement two adoption models. In both models, at each round, agents have access to the responses in the previous round, based on which they may improve their correctness in the new responses. Each agent also has a parameter $q_i$ that determines their openness to change its answer when exposed to a correct answer.
The actions are influenced using the following model. \textbf{(1) All influence (Global View):} For each question, if any agent in the previous round was correct, then every previously-incorrect agent in $S$ adopts the correct label in the next round with probability $q_i$. This models a fast, fully-mixed influence regime where each agent has the capacity to pay attention to all other agents.
\textbf{(2) Peer influence (Local View):} Each agent samples a small set of peers every round. An agent becomes correct with a probability $q_i$ in round $r$ if at least one of its sampled peers was correct in round $r-1$. This models the influence regime, when agents only pay attention to outputs of some agents.


\paragraph{Objective is monotone submodular.} The objective of maximizing correct answers is monotone submodular. Consider one instance of the flow of correct answers across agents obtained by sampling interactions, and whether an agent will correct itself if exposed to a correct answer. Agent-0 eventually has the correct answer if there is a path from any agent who had the correct answer to Agent-0. The monotonicity follows from the fact that the flow of correct answers is ``progressive''. That is, if we add another agent, it can only increase the number of paths initiating from a correct answer. Additionally, when we add an agent, it introduces new paths for the correct answer to flow to Agent-0 in the final round. Since reachability in a graph is submodular~\cite{kempe2003maximizing}, so is the correctness of Agent-0 in this instance. Finally, the expected outcome over the probability of correctness is a linear combination of such instances. Therefore, the expected outcome is also monotone submodular.

\paragraph{Cost is supermodular with bounded curvature.} The cost is a combination of input tokens and output tokens (at intermediate rounds and the final round). Consider two disjoint sets of agents $S$ and $T$, and an agent $i$. When we include $i$ along with $S \cup T$, it introduces exchange of information with $|S| + |T|$ agents. On the other hand, adding it to $|S|$, it introduces new exchange of information with only $|S|$ agents. As a result, $g(\{i\} \cup S \cup T) - g(S \cup T) \geq g(\{i\} \cup S ) - g(S)$. Therefore, the cost is supermodular. Further, assume, without loss of generality, $g(S) \geq g(T)$. Also, to simplify the analysis, suppose each agent outputs the same amount of tokens $O$. Then the additional cost introduced by adding set $S$ to $T$ can be shown to be $g(S) + (r-1) O(C_S^{O} + C_T^{O}) + O(C_S^{I} + C_T^{I})$. That is, cost within $S$ alone across $r$ rounds, and the cost of exchanges between $S$ and $T$. Here, $C$ indicates sum of costs of tokens for agents in $S$ (or $T$) and superscript $I$ and $O$ denote input and output, respectively. The net cost of $S$, i.e., $g(S)$ can be shown to be $> O C_S^{O} + I C_S^{I}$. This fact along with $g(S) \geq g(T)$ leads to $1-\kappa   = \frac{g(S)}{g(S|T)} \geq \frac{1}{3} \implies \kappa \leq \frac{2}{3}.$
A special case that is easy to calculate is when all costs are also the same. Then the $g(S) \propto |S|^2 + b|S|$ for some constant $b>0$. Then, using $g(S) \geq g(T) \implies |S|\geq |T|$, and so,
\begin{align*}
    1-\kappa   &= \frac{g(S)}{g(S|T)} = \frac{|S|^2 + b|S|}{(|S|+|T|)^2 + b(|S|+|T|) - |T|^2 - b|T|}  \\
    & = \frac{|S|^2 + b|S|}{|S|^2 + 2|T||S| + b|S|} 
     \geq \frac{|S|^2 + b|S|}{|S|^2 + 2|S|^2 + b|S|} \\
    & \geq \frac{|S| + b}{3|S| + b} \geq \frac{1}{3}
    \implies \kappa \leq \frac{2}{3}.
\end{align*}



\begin{figure}[!htbp]
    \centering
       \begin{subfigure}[b]{0.23\textwidth}
        \centering
        \includegraphics[width=\linewidth]{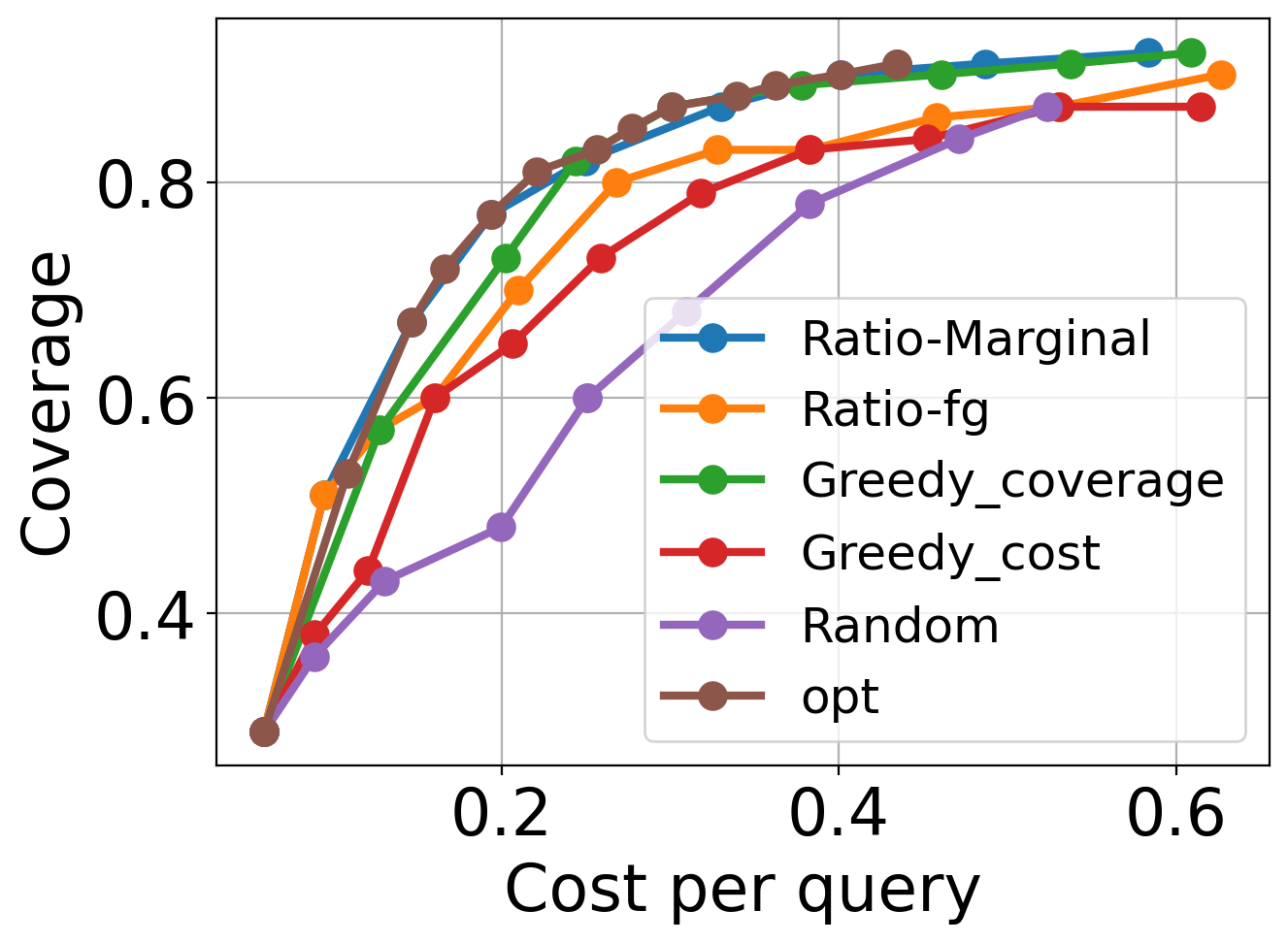}
        \caption{Global View, two rounds}
        \label{fig:sub1}
    \end{subfigure}
    \begin{subfigure}[b]{0.23\textwidth}
        \centering
        \includegraphics[width=\linewidth]{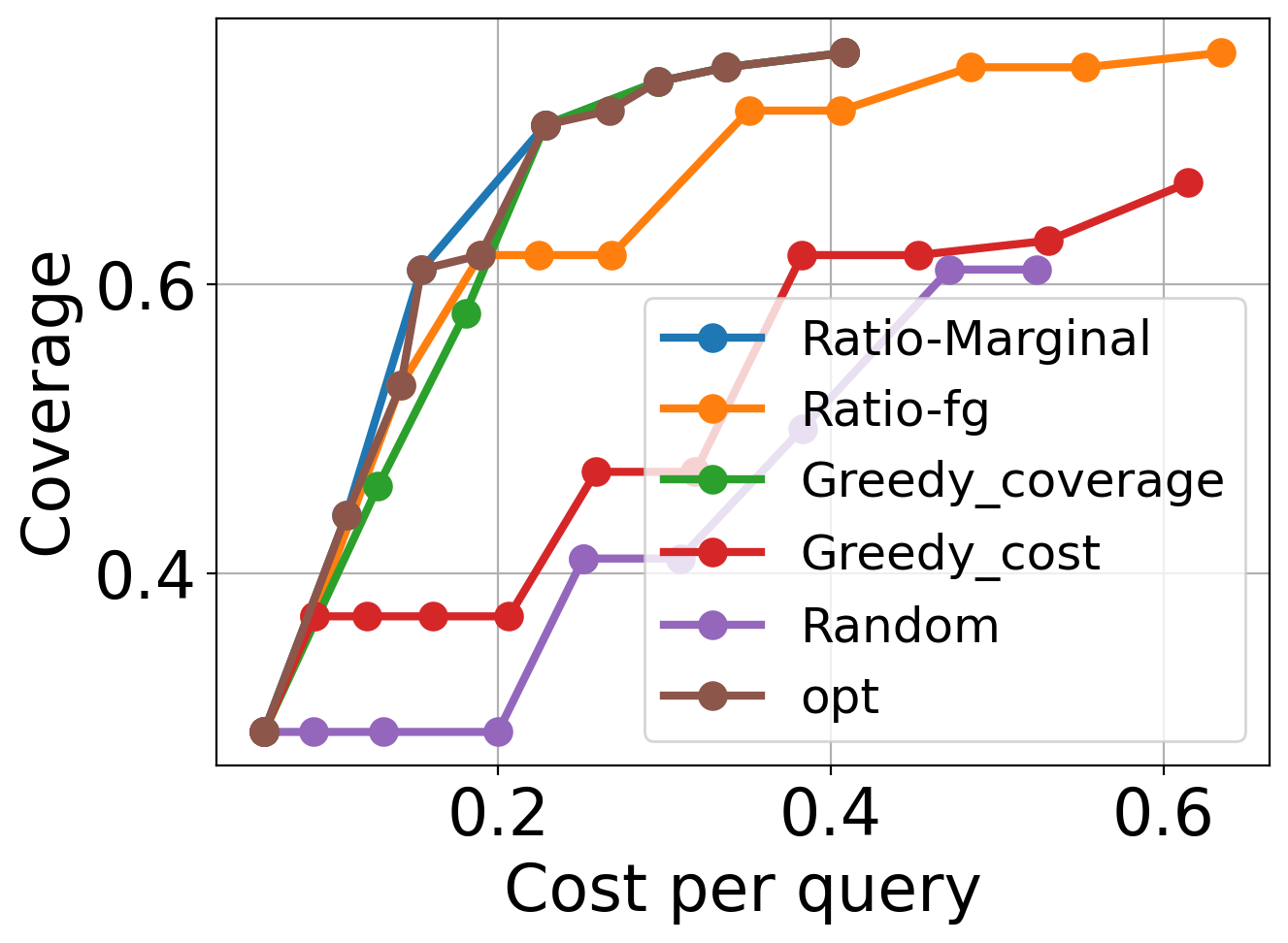}
        \caption{Local View, two rounds}
        \label{fig:sub1}
    \end{subfigure}
    
    \caption{Comparison including the optimal solution for 15 agents and 100 questions.}
    \label{fig:opt}
\end{figure}

\subsection{Algorithms Compared}

To empirically assess the performance of the proposed greedy approach with other heuristics, we compare the following methods. For each method, instead of solving the maximization problem at several cost cut-offs, we let the algorithms run to one fixed cost and plot $f$ vs $g$ for each iteration. 
\begin{itemize}
    \item Ratio-Marginal: The proposed greedy approach with approximation guarantees, which iteratively picks the highest ratio of marginal increases in $f$ and $g$.

    \item Ratio-fg: Another natural greedy choice, which iteratively picks based on the highest ratio of $f$ and $g$ at every step.

    \item Greedy-f: A greedy approach, which iteratively picks the highest marginal increases in $f$, ignoring $g$.
    
    \item Greedy-g: A greedy approach, which iteratively picks the smallest marginal increases in $g$, ignoring $f$.
    
    \item Random: Randomly picks an element to add to the set.

\end{itemize}

\paragraph{Implementation notes.} To make algorithm comparisons deterministic, we use identical pre-sampled interactions among layers.

\subsection{Results}
Figure~\ref{fig:global} shows the results of coverage vs cost for one of the instances over one, two, and three rounds for the Global View model. In all cases, Greedy-Ratio-Marginal outperforms all other heuristics. Ratio-fg is the closest competitor. On the other hand, for Local View model (Figure~\ref{fig:local}), the closest competitor is the Greedy-coverage heuristic for rounds one and two. Ratio-fg is significantly worse. We hypothesize that the other heuristics do not admit a good approximation ratio.
We also compare the greedy heuristics with exhaustive search (opt) for small problem sizes. We generate all possible subsets and evaluate their coverages and costs. For plotting, we set thresholds at small intervals for $g$ and pick the subset with the highest $f$. Figure~\ref{fig:opt} shows results for two rounds for the two debate models. We observe that Ratio-Marginal almost overlaps with the optimal solution\footnote{It may seem like sometimes Ratio-Marginal is better than opt, which is not true. These are actually scatter plots, as different algorithms lead to different costs. The line through the points only illustrates the trend.}. Greedy coverage comes close, while other heuristics are significantly worse. We ran several other instances, but the observations are consistent.

\paragraph{Other Application: Sensor Placement.} We ran another set of experiments for a sensor placement problem: $m$ sensors and $n$ targets to be covered. Each sensor has a different quality, resulting in different coverage and a different cost $c_i$.  The net cost of placing any set $S$ of $k$ out of $m$ sensors is $\sum_{c_i \in S} c_i +\lambda{k\choose 2}$, for some $\lambda$. The second term is the coordination cost across the sensors. It can be shown that this function is supermodular with curvature bounded by $2/3$. Figure~\ref{fig:sensor} shows the results. Similar to the debate simulation, for 15 sensors, we also calculate the optimal by exhaustive search. Results from Ratio-Marginal overlaps with the optimal. The advantage over Ratio-fg and Greedy\_coverage becomes more pronounced in the scenario with 100 sensors.

\begin{figure}[!htbp]
    \centering
    \begin{subfigure}[b]{0.23\textwidth}
        \centering
        \includegraphics[width=\linewidth]{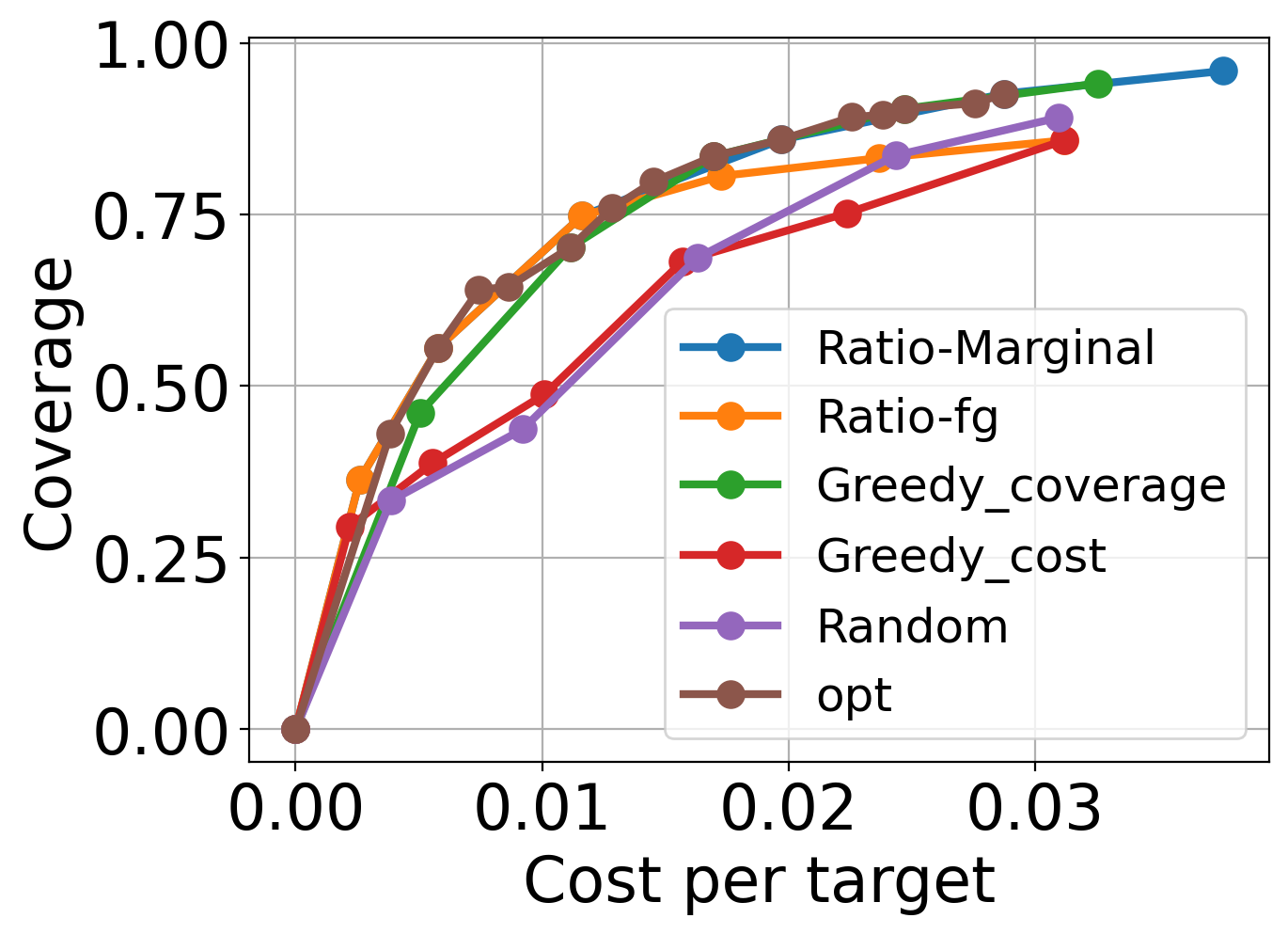}
        \caption{15 sensors}
        \label{fig:sub1}
    \end{subfigure}
    \hfill
    \begin{subfigure}[b]{0.23\textwidth}
        \centering
        \includegraphics[width=\linewidth]{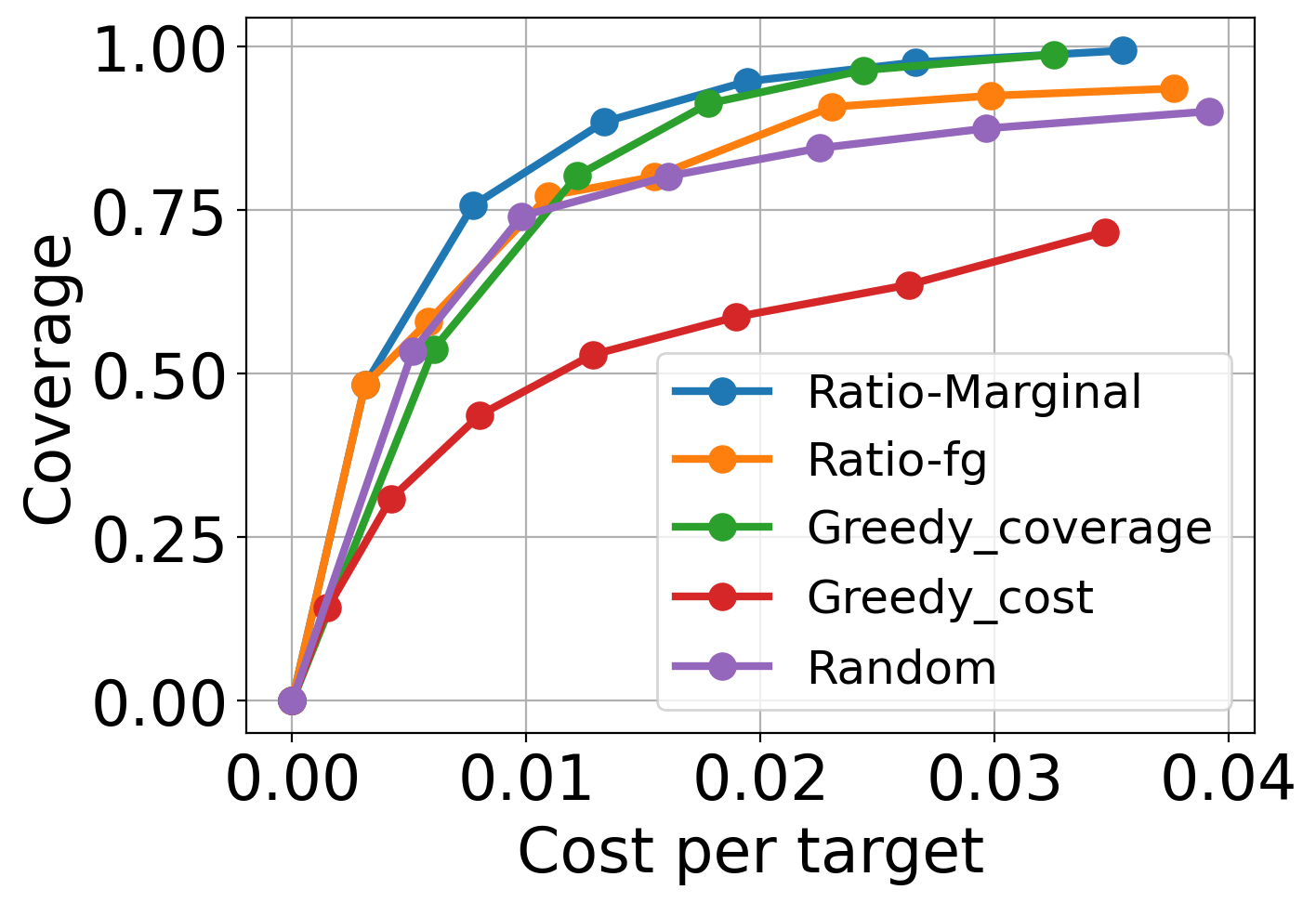}
        \caption{100 sensors}
        \label{fig:sub1}
    \end{subfigure}  
    \caption{Coverage vs cost for sensor placement.}
    \label{fig:sensor}
\end{figure}

\section{Conclusions and Remarks}

We study the performance of a greedy algorithm for maximizing a monotone submodular function subject to a monotone supermodular upper-bound constraint. We show that the problem is inapproximable in polynomial time in general, and therefore focus on the setting in which the supermodular cost function has bounded curvature. 
The proposed algorithm incrementally selects elements that maximize the ratio between the marginal gain in the objective and the marginal increase in the constraint. Despite its simplicity, we show that this greedy rule achieves constant-factor bicriteria approximation guarantees, with performance depending explicitly on the curvature parameters and the value of the constraint function at each step of the greedy process.
We also construct an instance to demonstrate the tightness of our performance characterization.

As a consequence of the tightness of our analysis, our results also recover and extend several classical special cases. For example, in the knapsack problem, both the objective and the cost functions are modular and therefore have zero curvature. Taking the limit as \(c\) and \(\gamma\) approach zero in Theorem~\ref{thm:curv-f}, we obtain an exact approximation ratio of \(1\), with a worst-case constraint violation by a factor of \(2\). When the objective function is submodular and the constraint is modular, Corollary~\ref{cor:overflow-traditional} yields a \((1 - e^{-\beta}, \beta)\) bicriteria approximation, matching known results~\cite{feldman2025bicriteria} by setting \(\epsilon = e^{-\beta}\).

Finally, as shown in Section \ref{sec:sub-sup}, our new and less restrictive notion of supermodular curvature applies to a broader range of settings.
In particular, a quadratic cost function has curvature \(1/3\), while a cost function proportional to the \(p\)-th power of the sum of individual costs has curvature \(\frac{2^p - 2}{2^p - 1}\). These examples illustrate that many natural supermodular cost functions admit bounded curvature, reinforcing the practical relevance of our framework.

\begin{acks}
    This work was supported in part by the National Science Foundation grant CCF-2308744.
\end{acks}


\clearpage
\bibliographystyle{ACM-Reference-Format}
\balance
\bibliography{refs}



\appendix
\input{short-appendix.tex}

\end{document}

%% file: short-appendix.tex
\section{Complete Proofs}

\subsection{With Curvature of $f$}

To analyze the performance of Greedy, for the case when $f$ has a curvature $c$, we will first derive a recurrence relation below.

\begin{lemma}\label{lem:curv_recur}
If $S_i$ is the set returned by greedy after $i$ steps (before $g(S_{i}$ overflows the constraint).
Suppose $f_i = f(S_i)$, $g_i = g(S_i)$ and $\Delta f_{i+1}$ and $\Delta g_{i+1}$ are the marginal gains in objective and the cost, respectively, obtained in step $i+1$, then
\[
(g^* - (1-c)(1-\gamma) g_i) \Delta f_{i+1} \geq (f^* - f_i)(1-\gamma) \Delta g_{i+1}.
\]
\end{lemma}

\begin{proof}


We define the following components of $S_i$:
\begin{itemize}
    \item $\hat{g}_i = g(S_i \cap S^*)$ (cost of elements in $S_i$ that are also in $S^*$)
    \item $\bar{g}_i = g(S_i \setminus S^*)$ (cost of elements in $S_i$ not in $S^*$)
    \item $\bar{f}_i = f(S_i \setminus S^*)$ (value of elements in $S_i$ not in $S^*$)
\end{itemize}

\noindent\textbf{Bounding $\bar{f}_i$ using Marginal Gains:}
Let $v_j$ be the element added at iteration $j$. Let $\delta_j = 1$ if $v_j \in S_i \setminus S^*$, and $\delta_j = 0$ otherwise. We consider the marginal gain of $\bar{f}_i$ at step $j$: $\Delta \bar{f}_j = \bar{f}_{j} - \bar{f}_{j-1}$.

By submodularity, for $j \leq i$:
\begin{equation}\label{eqn:fjbar2fj}
\Delta \bar{f}_j = f(v_j \mid S_{j-1} \setminus S^*) \geq \delta_j f(v_j \mid S_{j-1}) = \delta_j \Delta f_j
\end{equation}
If any $v_j \in S_i \setminus S^*$ was the best choice at step $j \leq i$, we have:
\begin{equation}\label{eqn:fj2gj}
\frac{\Delta f_j}{\Delta g_j} \geq \frac{f(v_{i+1}|S_j-1)}{g(v_{i+1}|S_{j-1})} \geq  \frac{\Delta f_{i+1}}{\Delta g_{i+1}} \implies \Delta f_j \geq \frac{\Delta f_{i+1}}{\Delta g_{i+1}} \Delta g_j\,,
\end{equation}
where the first inequality follows from greedy choice and the second follows from submodularity of the numerator and supermodularity of the denominator. Using Inequalities~\ref{eqn:fjbar2fj} and~\ref{eqn:fj2gj}
\begin{equation}
\Delta \bar{f}_j \geq \delta_j \frac{\Delta f_{i+1}}{\Delta g_{i+1}} \Delta g_j\,.
\end{equation}
Also, whenever we add an element that was not in $S^*$, $\delta_j = 1$, and due to supermodularity $\delta_j \Delta g_j \geq \Delta \bar{g}_j$. Therefore,
\begin{equation}
    \Delta \bar{f}_j \geq \frac{\Delta f_{i+1}}{\Delta g_{i+1}} \Delta \bar{g}_j\,.
\end{equation}

Summing over $j=1$ to $i$:

\begin{equation}
\bar{f}_i \geq \frac{\Delta f_{i+1}}{\Delta g_{i+1}} \bar{g}_i
\end{equation}

Recall that the greedy condition leads to
\begin{equation}\label{eqn:greedy_inequality}
\frac{\Delta f_{i+1}}{\Delta g_{i+1}} \geq \frac{f(S^* \mid S_i)}{g(S^* \mid S_i)}.    
\end{equation}

\noindent\textbf{Bounding the Numerator $f(S^* \mid S_i)$:}
\begin{equation}
f(S^* \cup S_i) = f(S^* \cup S_i \setminus S^*) = f^* + f(S_i \setminus S^* | S^*) \geq f^* + (1-c)\bar{f}_i\,,
\end{equation}
where the inequality follows from the submodular curvature $c$. Using this,
\begin{equation}
f(S^* \mid S_i)  = f(S^* \cup S_i) - f_i \geq f^* - f_i + (1-c)\bar{f}_i\,.
\end{equation}
Substituting the derived lower bound for $\bar{f}_i$:
\begin{equation}
f(S^* \mid S_i) \geq f^* - f_i + (1-c)\frac{\Delta f_{i+1}}{\Delta g_{i+1}} \bar{g}_i
\end{equation}

\noindent\textbf{Bounding the Denominator $g(S^* \mid S_i)$:}
Using the supermodular curvature $\gamma$ (noting that $g(S^*) \geq g(S_i) \geq g(S_i\setminus S^*)$):
\begin{equation}\label{eqn:g_marginal_bound}
g(S^* \mid S_i) = g(S^* \mid S_i \setminus S^*) + g(S_i \setminus S^*) - g_i  \leq \frac{g^*}{1-\gamma} + \bar{g}_i - g_i\,.
\end{equation}
\noindent\textbf{Substitution: }
Substituting the bounds into Inequality~\ref{eqn:greedy_inequality}:
\begin{equation}
\frac{\Delta f_{i+1}}{\Delta g_{i+1}} \geq \frac{f^* - f_i + (1-c)\frac{\Delta f_{i+1}}{\Delta g_{i+1}} \bar{g}_i}{\frac{g^*}{1-\gamma} + \bar{g}_i - g_i}
\end{equation}
Rearranging terms by isolating $\Delta f_{i+1}$:
\begin{equation}\label{eqn:recurrence_interm}
\Delta f_{i+1} \left[ \frac{g^*}{1-\gamma} + \bar{g}_i - g_i - (1-c) \bar{g}_i \right] \geq (f^* - f_i) \Delta g_{i+1}
\end{equation}
Simplifying the term in brackets:
\begin{align}
 \frac{g^*}{1-\gamma} + \bar{g}_i - g_i - (1-c) \bar{g}_i 
 &= \frac{g^*}{1-\gamma} + c\bar{g}_i - g_i  \nonumber\\ 
 &\leq  \frac{g^*}{1-\gamma} - (1-c)g_i \label{eqn:g_simplify_interm} \\
 &= \frac{g^* - (1-c)(1-\gamma)g_i}{1-\gamma}. \label{eqn:g_simplify}
\end{align}
where Inequality \label{eqn:g_simplify_interm} follows from $\bar{g_i} \leq g_i$.

\noindent\textbf{Final Recurrence Relation:}
Combining Inequalities~\ref{eqn:g_simplify} and~\ref{eqn:recurrence_interm}
\begin{equation}
\Delta f_{i+1} \left[ \frac{g^* - (1-c)(1-\gamma) g_i}{1 - \gamma} \right] \geq (f^* - f_i) \Delta g_{i+1}
\end{equation}
Rearranging:
\begin{equation}
(g^* - (1-c)(1-\gamma) g_i) \Delta f_{i+1} \geq (f^* - f_i)(1-\gamma) \Delta g_{i+1}
\end{equation}
\end{proof}

\input{extra-appendix}

\subsection{The Dual Problem}

\thmDual*

\begin{proof}
Let $B^*$ be the optimal dual value, i.e.,
\[
B^* = \min_{T:\, f(T)\ge\tau} g(T).
\]
Run the primal approximation algorithm $\mathcal{A}(B)$ at two budgets $b$ and $b-\varepsilon$ to obtain
\[
S_1 = \mathcal{A}(b), \qquad S_2 = \mathcal{A}(b-\varepsilon),
\]
such that
\[
f(S_1) \ge \alpha\tau, \qquad f(S_2) < \alpha\tau.
\]
From the primal guarantee for $S_2$, we have
\[
f(S_2) \ge \alpha\cdot \mathrm{OPT}_P(b-\varepsilon),
\]
which implies $\mathrm{OPT}_P(b-\varepsilon) < \tau$.
Hence, the budget $b-\varepsilon$ is insufficient to achieve $f(T)\ge\tau$, meaning
\[
b-\varepsilon < B^*.
\]
Therefore,
\[
b < B^* + \varepsilon.
\]

Using the primal cost guarantee at budget $b$, we have
\[
g(S_1) \le \beta b < \beta (B^* + \varepsilon)
= \beta\Big(1+\frac{\varepsilon}{B^*}\Big)B^*.
\]
Thus, the set $S_1$ satisfies
\[
f(S_1)\ge\alpha\tau, \qquad g(S_1)\le \beta\Big(1+\frac{\varepsilon}{B^*}\Big)B^*,
\]
which gives a $(\beta(1+\varepsilon/B^*),\alpha)$-approximation for the dual problem.
\end{proof}

%% file: extra-appendix.tex
To solve this recurrence relation, we first state and proof a Lemma that will help us transform one step of the recurrence.

\begin{lemma}\label{lem:interm_recurrence_step}
Let $q \geq  p \geq 1$. Suppose $a_i < q$, and $a_{i+1} \geq a_i$ follows
$p \frac{a_{i+1} - a_i}{q - a_i} < 1$. Then, 
\[
1 - p \frac{a_{i+1} - a_i}{q - a_i} \leq \left( \frac{q - a_{i+1}}{q - a_i} \right)^p
\]
\end{lemma}

\begin{proof}
First, we note that we must have $a_{i+1} \leq q$. If not then,

\[
p \frac{a_{i+1} - a_i}{q - a_i} \geq p \frac{q - a_i}{q - a_i} = 1,
\]
which violates the given condition. 
Now, let $x = 1 - \frac{q - a_{i+1}}{q - a_i}$. Note that:
\[
x = \frac{(q - a_i) - (q - a_{i+1})}{q - a_i} = \frac{a_{i+1} - a_i}{q - a_i}
\]
Since $q \geq a_{i+1} \geq a_i$, it follows that $0 \leq x \leq 1$. Applying Bernoulli's Inequality  which states that $(1 - x)^p \geq 1 - px$ for $x \in [0, 1]$ and $p \geq 1$, we have:

\begin{align*}
1 - p \left( \frac{a_{i+1} - a_i}{q - a_i} \right) &= 1 - p \left( \frac{(q - a_i) - (q - a_{i+1})}{q - a_i} \right) \\
&= 1 - p \left( 1 - \frac{q - a_{i+1}}{q - a_i} \right) \\
&\leq \left( 1 - \left( 1 - \frac{q - a_{i+1}}{q - a_i} \right) \right)^p = \left( \frac{q - a_{i+1}}{q - a_i} \right)^p
\end{align*}
\end{proof}

Now, we are ready to prove Theorem~\ref{thm:curv-f}.


\thmCurvF*

\begin{proof}
The discrete recurrence relation from Lemma~\ref{lem:curv_recur} is:
\begin{equation}
(g^* - (1-c)(1-\gamma) g_i) \Delta f_{i+1} \geq (f^* - f_i)(1-\gamma) \Delta g_{i+1}
\end{equation}
Dividing by $f^* g^*$ and rearranging terms gives the normalized recurrence:
\begin{align}
\frac{\Delta f_{i+1}}{f^*} &\geq \frac{f^* - f_i}{f^*} \frac{(1-\gamma) \Delta g_{i+1}}{g^* - (1-c)(1-\gamma) g_i}   \nonumber \\
\implies  F_{i+1} - F_i &\geq (1-F_i) \frac{(1-\gamma) \Delta G_{i+1}}{1 - (1-c)(1-\gamma) G_i} \nonumber \\
\implies  F_{i+1}  &\geq F_i + (1-F_i) \frac{(1-\gamma) \Delta G_{i+1}}{1 - (1-c)(1-\gamma) G_i}
\label{eq:norm_recurrence}
\end{align}
Here, if $\frac{(1-\gamma) \Delta G_{i+1}}{1 - (1-c)(1-\gamma) G_i} \geq 1$, we will get $F_{i+1} \geq F_i + (1-F_i) = 1$, i.e., $f_{i+1} \geq f^*$. Therefore, for a lower bound on $F_{i+1}$, we only need to consider the case when
\[
\frac{(1-\gamma) \Delta G_{i+1}}{1 - (1-c)(1-\gamma) G_i} < 1
\]

We rewrite Inequality~\ref{eq:norm_recurrence} as
\begin{align*}
(1 - F_{i+1}) &\geq (1-F_i) \left( 1 - \frac{(1-\gamma) \Delta G_{i+1}}{1 - (1-c)(1-\gamma) G_i} \right) \\
&= (1-F_i) \left( 1 - \frac{\frac{1}{1-c}\Delta G_{i+1}}{\frac{1}{(1-\gamma)(1-c)} -  G_i} \right)
\end{align*}

We apply Lemma~\ref{lem:interm_recurrence_step} with $p = \frac{1}{1-c} \leq q = \frac{1}{(1-c)(1-\gamma)}$ to get
\begin{align*}
(1-F_i) \left( \frac{q - G_{i+1}}{q - G_{i}} \right)^{p}\,.
\end{align*}
Recall that, due to the definition of our supermodular curvature, this relation is valid only up to $k = i+1$ when $G_{i+1}$ exceeds 1 for the first time\footnote{Here we are using $k$ as the index when we exceed $g(S^*)$ the first time rather than exceeding $\theta$}.
Telescoping from $i=0$ to $i+1 = k$, and using $F_i = 0, G_0 = 0$, we get

\begin{align}\label{eqn:approx-curv-f}
    1 - F_{k} &\leq \left( \frac{q - G_{k}}{q} \right)^{p} 
    \implies F_k &\geq 1 - \left(1 - (1-c)(1-\gamma)G_k \right)^\frac{1}{1-c}.
\end{align}
If $g_k \leq \theta$, then $G_k \geq 1$ was achieved before exceeding the constraint. In that case, due to the monotonicity of $f$, we already have
\[
F_k \geq 1 - \left(1 - (1-c)(1-\gamma) \right)^\frac{1}{1-c}\,.
\]
If $g_k > \theta$, then choose $\beta = \frac{g_k}{\theta} \geq \frac{g_k}{g^*} = G_k$, and so
\begin{equation}
    F_k \geq 1 - \left(1 - (1-c)(1-\gamma)\beta \right)^\frac{1}{1-c}\,.
\end{equation}
\end{proof}

\subsection{Continuing the Overflow}
Now, we analyze what happens if we continue to pick elements even after constraint violation. We first prove a lemma that will help us transform one step of a recurrence that we will derive for $i \geq k$.

\begin{lemma}\label{lem:over_recurrence_step}
Let $a_{i+1} > a_i$ such that
$p' \frac{a_{i+1} - a_i}{q' + a_i} < 1$ for some $p', q' > 0$. Then, 
:
\[
1 - p' \frac{a_{i+1} - a_i}{q' + a_i} \leq \left( \frac{q' + a_{i+1}}{q' + a_i} \right)^{p'}
\]
\end{lemma}

\begin{proof}
Let $x = \frac{a_{i+1} - a_i}{q' + a_i}$. Since $a_{i+1} > a_i$ and $q' > 0$, we have $x > 0$. We can rewrite the right-hand side of the inequality as:
\begin{align*}
\left( \frac{q' + a_{i+1}}{q' + a_i} \right)^{-p'} &= \left( \frac{q' + a_i + (a_{i+1} - a_i)}{q' + a_i} \right)^{-p'} = \left( 1 + \frac{a_{i+1} - a_i}{q' + a_i} \right)^{-p'} \\
&= (1 + x)^{-p'}
\end{align*}
We use the generalized Bernoulli Inequality, which states that for $x > -1$ and $r \leq 0$ (or $r \geq 1$):
\[
(1 + x)^r \geq 1 + rx
\]
Applying this with $r = -p'$ (where $-p' < 0$ since $p' > 0$):
\[
(1 + x)^{-p'} \geq 1 + (-p')x = 1 - p'x
\]
Substituting $x$ back into the inequality:
\[
\left( \frac{q' + a_{i+1}}{q' + a_i} \right)^{-p'} \geq 1 - p' \frac{a_{i+1} - a_i}{q' + a_i}
\]
\end{proof}

\thmBoundBeyond*

\begin{proof}
    Once we exceed the constraint, at step $i \geq k$, $g_i > g^*$. Then, Equation~\ref{eqn:g_marginal_bound} can be modified as
\begin{align}
    g(S^*|S_i) &= g(S_i|S^*) + g^* - g_i \nonumber \\
    &\leq \frac{g_i}{1-d} + g^* - g_i = \frac{d}{1-d}g_i + g^*\,.
\end{align}
Using this as in proof of Theorem~\ref{thm:curv-f} along with $\bar{g}_i \geq 0$, we get an updated recurrence relation for $i \geq k$:
\begin{equation}
    \Delta f_{i+1} ((1-\gamma)g^* + \gamma g_i) \geq (f^* - f_i)(1 - \gamma)\Delta g_{i+1}\,,
\end{equation}
which leads to 
\begin{equation}\label{eqn:overflow_interm}
    F_{i+1} \geq F_i + (1- F_i) \frac{\frac{1-\gamma}{\gamma}\Delta G_{i+1}}{\frac{1-\gamma}{\gamma} + G_i}\,.
\end{equation}
Again, if we have $\frac{\frac{1-\gamma}{\gamma}\Delta G_{i+1}}{\frac{1-\gamma}{\gamma} + G_i} \geq 1$, we already have $F_{i+1} \geq 1$. So we only need to consider the case where 
$
\frac{\frac{1-\gamma}{\gamma}\Delta G_{i+1}}{\frac{1-\gamma}{\gamma} + G_i} < 1\,.
$
We can rewrite Inequality~\ref{eqn:overflow_interm} as
\begin{equation}
    1 - F_{i+1} \leq (1-F_i)\left(1 - \frac{\frac{1-\gamma}{\gamma}\Delta G_{i+1}}{\frac{1-\gamma}{\gamma} + G_i}\right)\,.
\end{equation}
Applying Lemma~\ref{lem:over_recurrence_step} with $p' = q' = \frac{1-\gamma}{\gamma} > 0$, we get

\begin{equation}
    1 - F_{i+1} \leq (1-F_i)\left(\frac{\frac{1-\gamma}{\gamma} + G_{i+1}}{\frac{1-\gamma}{\gamma} + G_i}\right)^{-\frac{1-\gamma}{\gamma}}\,.
\end{equation}
Telescoping from $i = k$ to $i+1 = K$, for any $K > k$,
\begin{align}
    1 - F_K &\leq (1-F_k) \left(\frac{\frac{1-\gamma}{\gamma} + G_{K}}{\frac{1-\gamma}{\gamma} + G_k}\right)^{-\frac{1-\gamma}{\gamma}} \nonumber \\
    &\leq  \left(1-(1-c)(1-\gamma)G_k\right)^{\frac{1}{1-c}}\left(\frac{\frac{1-\gamma}{\gamma} + G_{K}}{\frac{1-\gamma}{\gamma} + G_k}\right)^{-\frac{1-\gamma}{\gamma}}
    \label{eqn:most_general}
\end{align}

Note that if we pick $k$ as the first time $g(S_k)$ exceeds $g(S^*)$, $G_k > 1$ and $G_k \leq \frac{2-d}{1-d}$. Also, $G_k \leq \frac{1}{(1-c)(1-\gamma)}$, otherwise, in Inequality~\ref{eqn:approx-curv-f}, we will have a negative value on the RHS greater than the LHS, which is only possible if $f_k > f^*$. In this range, the RHS of Inequality~\ref{eqn:most_general} is a decreasing function of $G_k$. Therefore, using $G_k > 1$, we get

\begin{equation}
    1 - F_K \leq  \left(1-(1-c)(1-\gamma) \right)^{\frac{1}{1-c}}\left(\frac{\frac{1-\gamma}{\gamma} + G_{K}}{\frac{1-\gamma}{\gamma} + 1}\right)^{-\frac{1-\gamma}{\gamma}}\,.
\end{equation}

Finally, let $\beta^+ = g_K/\theta$. Then, $G_K = g_K/g^* \geq g_K/\theta$. And, the RHS is a decreasing function of $G_K$, so
\begin{equation}
    1 - F_K \leq  \left(1-(1-c)(1-\gamma) \right)^{\frac{1}{1-c}}\left(\frac{\frac{1-\gamma}{\gamma} + \beta^+}{\frac{1-\gamma}{\gamma} + 1}\right)^{-\frac{1-\gamma}{\gamma}}\,.
\end{equation}
\end{proof}

%% file: refs.bib
@article{SHI2025113,
title = {Submodular + Supermodular function maximization with knapsack constraint},
journal = {Discrete Applied Mathematics},
volume = {377},
pages = {113-133},
year = {2025},
issn = {0166-218X},
doi = {https://doi.org/10.1016/j.dam.2025.06.062},
url = {https://www.sciencedirect.com/science/article/pii/S0166218X25003786},
author = {Majun Shi and Zishen Yang and Wei Wang}
}

@inbook{ModularityCurvatureSviridenko,
author = {Maxim Sviridenko and Jan Vondrák and Justin Ward},
title = {Optimal approximation for submodular and supermodular optimization with bounded curvature},
booktitle = {Proceedings of the 2015 Annual ACM-SIAM Symposium on Discrete Algorithms (SODA)},
chapter = {},
pages = {1134-1148},
doi = {10.1137/1.9781611973730.76},
URL = {https://epubs.siam.org/doi/abs/10.1137/1.9781611973730.76},
eprint = {https://epubs.siam.org/doi/pdf/10.1137/1.9781611973730.76}
}

@article{feldman2025bicriteria,
  title={Bicriteria submodular maximization},
  author={Feldman, Moran and Kuhnle, Alan},
  journal={arXiv preprint arXiv:2507.10248},
  year={2025}
}

@inproceedings{du2023improving,
  title={Improving factuality and reasoning in language models through multiagent debate},
  author={Du, Yilun and Li, Shuang and Torralba, Antonio and Tenenbaum, Joshua B and Mordatch, Igor},
  booktitle={Forty-first International Conference on Machine Learning},
  year={2023}
}

@inproceedings{bai2018greed,
  title={Greed is still good: maximizing monotone submodular+ supermodular (BP) functions},
  author={Bai, Wenruo and Bilmes, Jeff},
  booktitle={International Conference on Machine Learning},
  pages={304--313},
  year={2018},
  organization={PMLR}
}

@article{chen2023bicriteria,
  title={Bicriteria approximation algorithms for the submodular cover problem},
  author={Chen, Wenjing and Crawford, Victoria},
  journal={Advances in Neural Information Processing Systems},
  volume={36},
  pages={72705--72716},
  year={2023}
}

@inproceedings{padmanabhan2023maximizing,
  title={Maximizing submodular functions under submodular constraints},
  author={Padmanabhan, Madhavan R and Zhu, Yanhui and Basu, Samik and Pavan, Aduri},
  booktitle={Uncertainty in Artificial Intelligence},
  pages={1618--1627},
  year={2023},
  organization={PMLR}
}

@article{nemhauser1978analysis,
  title={An analysis of approximations for maximizing submodular set functions—I},
  author={Nemhauser, George L and Wolsey, Laurence A and Fisher, Marshall L},
  journal={Mathematical programming},
  volume={14},
  number={1},
  pages={265--294},
  year={1978},
  publisher={Springer}
}

@inproceedings{kempe2003maximizing,
  title={Maximizing the spread of influence through a social network},
  author={Kempe, David and Kleinberg, Jon and Tardos, {\'E}va},
  booktitle={Proceedings of the ninth ACM SIGKDD international conference on Knowledge discovery and data mining},
  pages={137--146},
  year={2003}
}

@article{yang2021learning,
  title={Learning interpretable decision rule sets: A submodular optimization approach},
  author={Yang, Fan and He, Kai and Yang, Linxiao and Du, Hongxia and Yang, Jingbang and Yang, Bo and Sun, Liang},
  journal={Advances in Neural Information Processing Systems},
  volume={34},
  pages={27890--27902},
  year={2021}
}

@article{chen2024less,
  title={Less is more: Fewer interpretable region via submodular subset selection},
  author={Chen, Ruoyu and Zhang, Hua and Liang, Siyuan and Li, Jingzhi and Cao, Xiaochun},
  journal={arXiv preprint arXiv:2402.09164},
  year={2024}
}

@inproceedings{kothawade2022prism,
  title={Prism: A rich class of parameterized submodular information measures for guided data subset selection},
  author={Kothawade, Suraj and Kaushal, Vishal and Ramakrishnan, Ganesh and Bilmes, Jeff and Iyer, Rishabh},
  booktitle={Proceedings of the AAAI Conference on Artificial Intelligence},
  volume={36},
  number={9},
  pages={10238--10246},
  year={2022}
}

@article{zhao2025dynaact,
  title={DynaAct: Large Language Model Reasoning with Dynamic Action Spaces},
  author={Zhao, Xueliang and Wu, Wei and Guan, Jian and Li, Qintong and Kong, Lingpeng},
  journal={arXiv preprint arXiv:2511.08043},
  year={2025}
}

@article{islam2021evaluation,
  title={Evaluation of the United States COVID-19 vaccine allocation strategy},
  author={Islam, Md Rafiul and Oraby, Tamer and McCombs, Audrey and Chowdhury, Mohammad Mihrab and Al-Mamun, Mohammad and Tyshenko, Michael G and Kadelka, Claus},
  journal={PloS one},
  volume={16},
  number={11},
  pages={e0259700},
  year={2021},
  publisher={Public Library of Science San Francisco, CA USA}
}

@article{chandrashekar2014survey,
  title={A survey on feature selection methods},
  author={Chandrashekar, Girish and Sahin, Ferat},
  journal={Computers \& electrical engineering},
  volume={40},
  number={1},
  pages={16--28},
  year={2014},
  publisher={Elsevier}
}

@inproceedings{chekuri2022densest,
  title={Densest subgraph: Supermodularity, iterative peeling, and flow},
  author={Chekuri, Chandra and Quanrud, Kent and Torres, Manuel R},
  booktitle={Proceedings of the 2022 Annual ACM-SIAM Symposium on Discrete Algorithms (SODA)},
  pages={1531--1555},
  year={2022},
  organization={SIAM}
}

@article{haastad1999clique,
  title={Clique is hard to approximate within n 1- $\varepsilon$},
  author={H{\aa}stad, Johan},
  year={1999}
}
